\documentclass{article}
\usepackage{amsmath}
\usepackage{amsthm}
\usepackage{amsfonts}
\usepackage{amssymb}
\usepackage{graphicx}
\usepackage{enumerate}
\usepackage{a4wide}
\usepackage{pgfplots}
\usepackage{algorithm}
\usepackage{algpseudocode}
\usepackage{subcaption}
\usepackage{todonotes}
\usepackage{apxproof}
\usepackage{multirow}
\usepackage{xurl}
\usepackage{authblk}
\newtheorem{theorem}{Theorem}[section]
\newtheorem{lemma}[theorem]{Lemma}
\newtheorem{corollary}[theorem]{Corollary}
\newtheorem{proposition}[theorem]{Proposition}
\theoremstyle{definition}
\newtheorem{definition}[theorem]{Definition}
\newtheorem{assumption}{Assumption}
\theoremstyle{remark}
\newtheorem{remark}[theorem]{Remark}
\newtheorem{example}[theorem]{Example}
\renewcommand{\algorithmicrequire}{\textbf{Input:}}
\renewcommand{\algorithmicensure}{\textbf{Output:}}
\newcommand{\E}{\mathbb{E}}
\newcommand{\R}{\mathbb{R}}
\newcommand{\N}{\mathbb{N}}
\newcommand{\J}{\mathbb{J}}

\newcommand{\bx}{\boldsymbol{x}}

\newcommand{\tr}{\operatorname{Tr}}

\newcommand{\argmin}{\operatorname*{argmin}}
\newcommand{\argmax}{\operatorname*{argmax}}

\newcommand{\dom}{{\mathcal X}}
\newcommand{\bmu}{\boldsymbol{\mu}}
\newcommand{\Pb}{{\mathcal P}^B(\dom)}
\newcommand{\bUb}{{\bf U}_B}

\newcommand{\norm}[1]{\left\|#1\right\|}

\title{Batch-based Bayesian Optimal Experimental Design in Linear Inverse Problems}
\author[$\ddagger$]{Sofia M\"akinen}
\author[$\dagger$]{Andrew Duncan}
\author[$\ddagger$]{Tapio Helin}
\affil[$\dagger$]{Imperial College London, UK}
\affil[$\ddagger$]{LUT University, Finland}
\date{\today}

\begin{document}
\maketitle
\begin{abstract}
Experimental design is central to science and engineering. A ubiquitous challenge is how to maximize the value of information obtained from expensive or constrained experimental settings. Bayesian optimal experimental design (OED) provides a principled framework for addressing such questions. In this paper, we study experimental design problems such as the optimization of sensor locations over a continuous domain in the context of linear Bayesian inverse problems. We focus in particular on batch design, that is, the simultaneous optimization of multiple design variables, which leads to a notoriously difficult non-convex optimization problem.
We tackle this challenge using a promising strategy recently proposed in the frequentist setting, which relaxes A-optimal design to the space of finite positive measures. Our main contribution is the rigorous identification of the Bayesian inference problem corresponding to this relaxed A-optimal OED formulation. Moreover, building on recent work, we develop a Wasserstein gradient-flow–based optimization algorithm for the expected utility and introduce novel regularization schemes that guarantee convergence to an empirical measure. These theoretical results are supported by numerical experiments demonstrating both convergence and the effectiveness of the proposed regularization strategy.
\end{abstract}

\section{Introduction}

The design of experiments is a central theme across science and engineering. Data collection is often restricted by practical constraints such as physical restrictions (consider radiation exposure in medical imaging) or high costs (consider seismic imaging), to name a few.
The allocation of measurement effort, the positioning of sensors, and the timing of observations are all decisions that shape the quality of inference one can draw from an experiment. This ubiquity has motivated the development of a wide range of approaches, from heuristic rules of thumb to formal optimization-based strategies, with the goal of extracting maximal value from every measurement.

A particularly principled framework for addressing such questions is Bayesian optimal experimental design (OED). In this approach, prior knowledge about the system under study is combined with a probabilistic model for future observations. Moreover, a utility function quantifies the expected benefit of making a particular set of measurements. Thus, Bayesian design naturally accounts for uncertainties in both parameters and data, allowing one to reason explicitly about information gain and decision quality. Standard choices for the optimality condition include the A- and D-optimality criteria, of which we consider the A-optimality \cite{review_cv_boed}. While exact computation of these quantities is often infeasible, advances in computation and approximation methods have enabled Bayesian design to be applied to increasingly complex systems.

In this work, we consider experimental design problems defined over a continuous domain. Our focus is on batch-based optimization, where multiple design variables, such as sensor locations, are optimized simultaneously. Such batch design problems are notoriously nonconvex and exhibit multiple local minima, which makes them challenging to solve reliably. A promising recent strategy to address these difficulties is to convexify the optimization problem by relaxing it to the space of probability measures. This relaxation has been successfully employed in related contexts, both in Bayesian optimization \cite{crovini2022batch} and in frequentist settings \cite{2ndBatchOptArticle, jin2024continuous, shi2026gradient}.

Our work is particularly closely aligned with the objectives studied in \cite{2ndBatchOptArticle}. Let us introduce some notation to make this connection precise. Suppose function $g$ in a reproducing kernel Hilbert space on the domain $\dom$
is observed at a point $x \in \dom$, connecting it to an unknown function of interest $f$ through the measurement model
\begin{equation}
\label{eq:informal_obs_model}
y(x) = (Af)(x) + \epsilon,
\end{equation}
where $A$ is a linear operator and $\epsilon$ is additive Gaussian noise with variance $\sigma^2$, independent of both $f$ and $x$.
In this setting, \cite{2ndBatchOptArticle} investigates, in particular, the frequentist A-optimal measurement strategy, which under suitable assumptions on the measurement model amounts to solving
\begin{equation}
\label{eq:informal_Aopt_freq}
x_{Aopt} \in \argmin_{x \in X} \tr \left[(A_x^*A_x)^{-1}\right],
\end{equation}
where $A_x$ denotes the operator $A$ followed by evaluation at $x$.
The main idea in \cite{2ndBatchOptArticle} is to relax the optimization task in \eqref{eq:informal_Aopt_freq} to the space of probability measures ${\mathcal P}(\dom)$ by informally setting
\begin{equation}
\label{eq:informal_Aopt_freq_meas}
\mu_{Aopt} \in \argmin_{\mu \in {\mathcal P}(\dom)} \tr \left[\left(\int_X A_z^*A_z  \mu(dz)\right)^{-1}\right],
\end{equation}
where the minimization functional can be shown to be convex with respect to the probability measure $\mu$. The relaxation can be understood by noting that the formulation \eqref{eq:informal_Aopt_freq_meas} coincides with \eqref{eq:informal_Aopt_freq} when the search for $\mu$ is restricted to Dirac measures $\{\delta_x \in {\mathcal P}(\dom) \; | \; x \in \dom\}$. The key contribution in \cite{2ndBatchOptArticle} is to formulate a Wasserstein gradient flow scheme for optimizing $\mu$ by particle approximations. 

The relaxation extends naturally to a batch-based setting. If one seeks to optimize a batch of sensor locations ${x_1,\dots,x_B}\subset \dom$ in \eqref{eq:informal_Aopt_freq}, then the corresponding object in \eqref{eq:informal_Aopt_freq_meas} is the measure $\sum_{j=1}^B \delta_{x_j}$. This suggests extending the relaxation to all positive measures satisfying $\mu(\dom)=B$. The resulting shift, from a highly non-convex search over $\otimes_{j=1}^B \dom$ to a convex optimization problem over measures on $\dom$, is precisely what makes this approach appealing for large-scale multi-parameter design problems.

To develop this approach further, in this paper, we study whether the experimental design setting of \eqref{eq:informal_Aopt_freq} and \eqref{eq:informal_Aopt_freq_meas} can be extended to a Bayesian framework. When the prior distribution of $f$ is Gaussian, the posterior induced by \eqref{eq:informal_obs_model} remains Gaussian. In this context, the A-optimality principle amounts to minimizing the posterior variance, however, a coherent relaxation of the problem and its interpretation in the Bayesian paradigm is not trivial. This observation motivates the central questions of our work: what observation does the relaxation correspond to in the Bayesian setting and is there a corresponding relaxation of the design problem aligned with \eqref{eq:informal_Aopt_freq_meas}?

\subsection{Our Contribution}

In this paper, we rigorously identify the Bayesian inference task associated to the frequentist batch-based A-optimal design problem proposed in \cite{2ndBatchOptArticle}, when relaxed to finite positive measures with fixed mass and interpreted through the Bayesian paradigm. 
This characterization is developed in Sections \ref{sec:relaxing_BIP} and \ref{sec:BOED}, where 
\begin{itemize}
    \item[(1)] we propose a non-parametric Bayesian inference task based on observing a Gaussian random process indexed over a Hilbert space associated to the relaxed design $\mu$, and,
    \item[(2)] we derive the corresponding posterior covariance, which leads to an infinite-dimensional expected utility functional in Section~\ref{subsec:Aopt}.
\end{itemize} 
Importantly, in Proposition \ref{prop:concavity} we prove that this expected utility is concave with respect to the design measure $\mu$.

Following previous work \cite{2ndBatchOptArticle}, we introduce a Wasserstein gradient flow for the particle-based optimization of the expected utility. To this end, we derive the first variation and the gradient of the expected utility in Propositions \ref{prop:frechet_derivative} and \ref{prop:gradient_of_U}. 

While the optimization method is developed for positive measures with fixed mass, another key contribution of this work is the proposal of an alternative approach to accelerate convergence. Specifically, each of the $B$ design variables is represented by a probability measure $\mu_j$, $j=1,\dots,B$, and the optimization problem is lifted back to $\dom^B$, that is, the optimization is carried out over the tensor product $\bmu = \bigotimes_{j=1}^B \mu_j$. This formulation enables us to \emph{regularize} the expected utility through interplay of the design measures $\mu_j$.

We introduce a two-fold regularization strategy, whose relative influence can be tuned by
reweighting:
\begin{itemize}
    \item[(A)] First, we penalize the expected utility by the empirical variance of each ensemble
    in order to ensure concentration toward a single-point measure.
    \item[(B)] Second, we penalize proximity between distinct ensembles via the negative maximum
    mean discrepancy, thereby discouraging multiple ensembles from collapsing to the same
    point measure.
\end{itemize}
While convergence to multiple point measures, or collapse to a single joint point
measure, may be justified by the inference task and the resulting information gain, these
regularization mechanisms provide practical tools for enforcing additional constraints
arising in experimental settings, such as prescribed minimum distances between sensor
locations. These effects are demonstrated through numerical examples in Section \ref{sec:numerical_simulations}.

\subsection{Literature overview}

Bayesian experimental design is a mature research area supported by a substantial and evolving literature. Recent survey-style contributions include \cite{Huan_Jagalur_Marzouk_2024, rainforth2024modern, ryan16}. This work focuses on A-optimality criterion which is a classical utility construction discussed e.g. in \cite{review_cv_boed}. We also highlight the recent work to promote Wasserstein distance based information criterion in \cite{helin2025bayesian}, motivated by the general notion of a valid information measure proposed by Ginebra in \cite{ginebra2007}.

Inverse problems form a class of high-dimensional inference tasks in which unknown parameters are linked to observations through complex forward models, often governed by partial differential equations. The requirement that experimental design criteria remain well defined across refinement levels has motivated extensions of classical Bayesian optimal experimental design to infinite-dimensional settings; for example, expected information gain and A-optimality criterion admit rigorous nonparametric formulations as developed in \cite{alexanderianAandDOptimalityInfDim}. See also the review \cite{review_aa}.

The joint optimization of multiple sensor locations prior to any data acquisition has been widely studied in the literature. Owing to the highly nonconvex nature of the resulting design problems, many existing approaches rely on greedy selection procedures, which consequently provide limited guarantees of global optimality. For recent contributions in the setting of general PDE-based inverse problems, see \cite{aretz2021sequential, go2025sequential, cui2025subspace}. In particular, Bayesian sensor placement for electrical impedance tomography has been investigated in \cite{hyvonen2024bayesian, karimi2021optimal}. This line of work should be distinguished from a different sequential paradigm in which data are collected after each newly selected sensor location, see e.g. \cite{helin2022edge, burger2021sequentially}.

Gradient-flow–based approaches to optimal experimental design in the frequentist setting have been investigated in \cite{2ndBatchOptArticle, jin2024continuous, shi2026gradient}, where the optimization problem is formulated over the space of probability measures. In particular, \cite{shi2026gradient} studies Wasserstein gradient flows for E-optimal design in regression models, while \cite{2ndBatchOptArticle} and \cite{jin2024continuous} examine A- and D-optimal continuous designs, focusing respectively on linear and nonlinear settings. Separately, \cite{crovini2022batch} explores the use of Wasserstein and Stein gradient flows for batch Bayesian optimization aimed at selecting multiple evaluation points in parallel. For more broad discussion on Wasserstein gradient flows, we refer the reader to \cite{gradientflow_book, gradFlowsOverview}. Optimization problems over probability measures via gradient flows have also been considered for general machine learning tasks, see \cite{chen2025accelerating}.\\

\noindent {\bf Structure of the Paper.} Section \ref{sec:relaxing_BIP} introduces the Bayesian inference task giving rise to the relaxed design problem over finite measures. The corresponding Bayesian OED task and the A-optimal expected utility are studied in Section \ref{sec:BOED}. Section \ref{sec:Wasserstein_gradient_flow} introduces the Wasserstein gradient flow for the optimization to the expected utility, while Section \ref{sec:regularization} is devoted to the practical regularization strategies in batch-based optimization. Finally, Section \ref{sec:numerical_simulations} illustrates the method for two inverse problems with one-dimensional Poisson and two-dimensional time-harmonic Schr\"odinger equation, with a sensitivity study regarding the regularization in Section \ref{subsec:sensitivity}.

\section{Relaxation of the Bayesian inference problem}
\label{sec:relaxing_BIP}

Let $H_k$ be a separable reproducing kernel Hilbert space on a bounded smooth domain $\dom \subset \R^n$, generated by a bounded kernel $k$ satisfying the required regularity conditions. Since the kernel remains fixed throughout the paper, we write $H = H_k$ for brevity. We consider optimizing the sensor location $x\in \dom$ for recovering $f\in U$ in a separable Hilbert space $U$ from
\begin{equation}
    g(x) = \langle K_x, A f\rangle_{H} + \epsilon\in \R
\end{equation}
where $A: U \to H$ is linear and bounded, $K_x\in H$ is the reproducing element and $\epsilon$ is normally distributed noise independent of the sensor location. 

To relax the design problem, let $\mu$ be a positive finite measure on $\dom$, and for brevity write $L^2(\mu) := L^2(\dom; \mu)$. Also, let us introduce the embedding operator $\iota_\mu: H_k \to L^2(\mu)$. In what follows, we study an inference problem, where the observation $g_\mu$ is a Gaussian random process indexed by $L^2(\mu)$ satisfying
\begin{equation}
    \label{eq:relaxed_inference}
    g_\mu(h) = \langle A_\mu f, h\rangle_{L^2(\mu)} + \epsilon_\mu(h),
\end{equation} 
where we write $A_\mu = \iota_\mu A : U \to L^2(\mu)$ and $\epsilon_\mu$ is a isonormal Gaussian process over $L^2(\mu)$ \cite[Def. 1.1.1]{nualart2006malliavin}.

The non-parametric Bayesian inference of the signal $f \in U$ is well-understood \cite{ghosal2017fundamentals}. In particular, for a Gaussian prior distribution with a linear Hilbert-space valued observational model, the posterior distribution was established in \cite{mandelbaum1984linear} and has been revisited in many subsequent work (see e.g. \cite{luschgy1996linear, hairer2005analysis}). While the random process formulation such as in \eqref{eq:relaxed_inference} with signal contaminated with white noise has received less attention, it is well-known that there exists a regular version of the conditional distribution of $f$ given $g_\mu$ and that this conditional distribution is Gaussian \cite{lehtinen1989linear}. 

For the reader’s convenience we recall a standard construction of the posterior covariance in Theorem \ref{thm:posterior_covariance1} that simplifies the setting of \cite{lehtinen1989linear}. Starting from the Hilbert space $L^2(\mu)$, one introduces a Hilbert scale $Z^\mu \subset L^2(\mu) \subset (Z^\mu)'$, where $Z^\mu$ is a separable Hilbert space that embeds continuously and densely into $L^2(\mu)$ with the property that realizations of $\epsilon_\mu$ are 'contained' in $(Z^\mu)'$. This procedure is a Hilbert-space version of the abstract Wiener space construction that is standard in modern Gaussian measure theory \cite{stroock2023gaussian}. 

More concretely, let $\{e_j\}_{j\in\J}$ be the orthonormal basis of $L^2(\mu)$, where $\J$ is either finite or countable. We define 
\begin{equation}
    Z_s^\mu = \left\{z = \sum_{j\in\J} z_j e_j \; \bigg| \; \sum_{j\in\J} j^{2s} z_j^2 < \infty \right\}, \quad s\in \R, 
\end{equation}
with the interpretation $Z_0^\mu = L^2(\mu)$ and $(Z_s^\mu)' = Z_{-s}^\mu$. For the moment, suppose $\J=\N$. We immediately observe that the random series $\sum_{j=1}^J \xi_j e_j$, $\xi_j\in{\mathcal N}(0,1)$ i.i.d., converges to a well-defined random variable $\tilde \epsilon_\mu$ in $L^2(\Omega; Z_{-s}^\mu)$ such that 
\begin{equation}
    \epsilon_\mu(z) = \langle \tilde \epsilon_\mu, z\rangle_{Z_{-s}^\mu\times Z_s^\mu},
\end{equation}
in distribution for any $z\in Z_s^\mu$, when $s$ is large enough, i.e., $\sum_{j\in\J} j^{-2s}<\infty.$ Clearly, for a finite index the claim is trivial.

We can now consider the joint distribution $(\tilde g_\mu, f)$ in the product space $Z_{-s}^\mu \times U$, where $\tilde g_\mu$ conditioned on $f$ has the distribution of $A_\mu f + \tilde \epsilon_\mu$ and $s$ is fixed large enough. With the abuse of notation, we will identify $g_\mu$ and $\tilde g_\mu$ (as well as, $\epsilon_\mu$ and $\tilde \epsilon_\mu$) in what follows. Moreover, below, we write $A_\mu^*$ for the adjoint of $A_\mu : U \to L^2(\mu)$, while $A_\mu'$ denotes the adjoint of $A_\mu : U \to Z_{-s}^\mu$, where the latter adjoint is obtained the canonical embedding of the range $L^2(\mu)$ into $Z_{-s}^\mu$.

\begin{theorem}
\label{thm:posterior_covariance1}
Suppose the random variable $(g_\mu, f) \in (Z_{-s}^\mu, U)$ has the joint distribution induced by the marginal $f\in {\mathcal N}(0, C_f)$ on $U$ and the likelihood induced by \eqref{eq:relaxed_inference}. Then the conditional distribution of $f$ given $g_\mu$ has a regular Gaussian version with a trace-class covariance operator
\begin{equation}
    \label{eq:posterior_covariance_form1}
    C_{post}^\mu = C_f - C_f A_\mu^* \left(A_\mu C_f A_\mu^* + I\right)^{-1} A_\mu C_f : U\to U.
\end{equation}
\end{theorem}

\begin{proof}
The regular conditional distribution of $f$ given $g_\mu$ is Gaussian with the covariance $C_{post}^\mu : U\to U$ given by
\begin{equation}
    \label{eq:posterior_covariance_general}
    C_{post}^\mu = C_f - C_{f, g_\mu} C_{g_\mu}^{-1}C_{g_\mu, f},
\end{equation}
where $C_{f,f'} = \E \left[f \otimes f'\right]$ stands for the cross-covariance of random variables $f$ and $f'$, and $C_g$ is the covariance operator of $g_\mu$ \cite{mandelbaum1984linear}.
While $C_g$ is not invertible in general, the product $C_{f, g_\mu} C_{g_\mu}^{-\frac 12}$ is well-defined as a Hilbert--Schmidt operator, see e.g. \cite[Lemma 4.2]{hairer2005analysis}.

It is straightforward to observe that 
\begin{equation*}
    C_{f, g_\mu} = C_f A_\mu' : Z_{-s}^\mu \to U \quad \text{and} \quad
    C_{g_\mu, f} = A_\mu C_f : U \to Z_{-s}^\mu.
\end{equation*}
Furthermore, due to the Hilbert scale construction, there exists a linear bounded bijection $R : Z_s^\mu \to Z_{-s}^\mu$ induced in general by the Riesz representation such that
\begin{equation}
    \langle a, b\rangle_{Z_{-s}^\mu \times Z_s^\mu} = \langle a, b \rangle_{Z_0^\mu}
    = \langle a, R b\rangle_{Z_{-s}^\mu},
\end{equation}
for any $a\in Z_0^\mu$ and $b\in Z_s^\mu$. It directly follows that $A_\mu' = A_\mu^* R^{-1} : Z_{-s}^\mu \to U$. 

Next, since $f$ and $\epsilon_\mu$ are independent, we have for any $a,b\in Z_{-s}^\mu$ that
\begin{eqnarray*}
    \langle C_g a, b\rangle_{Z_{-s}^\mu} & = & 
    \E \langle A_\mu u, a\rangle_{Z_{-s}^\mu} \langle A_\mu u, b\rangle_{Z_{-s}^\mu} 
    + \E \langle \epsilon_\mu, a\rangle_{Z_{-s}^\mu} \langle \epsilon_\mu, b\rangle_{Z_{-s}^\mu} \\
    & = & \langle C_f A_\mu' a, A_\mu' b\rangle_{U}
    + \E \langle \epsilon_\mu,  R^{-1} a\rangle_{Z_{-s}^\mu\times Z_s^\mu} \langle \epsilon_\mu, R^{-1} b\rangle_{Z_{-s}^\mu \times Z_s^\mu} \\
    & = & \langle A_\mu C_f A_\mu' a, b\rangle_{Z_{-s}^\mu} + \langle R^{-1} a, R^{-1} b\rangle_{Z_0^\mu} \\
    & = & \langle A_\mu C_f A_\mu' a, b\rangle_{Z_{-s}^\mu} + \langle R^{-1} a, b\rangle_{Z_{-s}^\mu}.
\end{eqnarray*}
In consequence, we have
\begin{equation}
    C_g = A_\mu C_f A_\mu' + R^{-1} : Z_{-s}^\mu \to Z_{-s}^\mu.
\end{equation}
Let us next show that the operator \eqref{eq:posterior_covariance_general} has the same action on $U$ as the operator \eqref{eq:posterior_covariance_form1}.
Clearly, in the latter operator, the inverse $\left(A_\mu C_f A_\mu^* + I\right)^{-1}$ is well-defined linear and bounded operator on $Z_0^\mu$.
To show that the actions coincide, let $u\in {\mathcal R}(A_\mu C_f) \subset Z_0^\mu$. It remains to show that 
\begin{equation*}
    R^{-1} \left[\left(A_\mu C_f A_\mu^* + I\right)R^{-1}\right]^{-1} u = \left(A_\mu C_f A_\mu^* + I\right)^{-1} u.
\end{equation*}
Notice that due to our Hilbert scale construction, $R$ is a diagonal operator in the orthonormal basis $\{e_j\}_{j\in\J}$, which extends to a bijective mapping between $Z_t^\mu$ and $Z_{t-2s}^\mu$ for any $t\in\R$. It follows that there exists a unique vector $w \in Z_{2s}^\mu$ such that
\begin{equation*}
    (A_\mu C_f A_\mu^* + I)R^{-1} w = u \in Z_0^\mu.
\end{equation*}
In consequence, we have
\begin{equation*}
    R^{-1}w = (A_\mu C_f A_\mu^* + I)^{-1} u
\end{equation*}
which proves the claim.
\end{proof}

To better characterize the self-adjoint operator $A_\mu^*A_\mu$, observe that the adjoint
$\iota_\mu^*:L^2(\dom;\mu)\to H$ satisfies
\begin{equation*}
    \langle \iota_\mu^* f, g\rangle_{H}
    = \langle f, \iota_\mu g\rangle_{L^2(\mu)}
    = \int_{\dom} f(x)\,\langle k_x,g\rangle_H\,\mu(dx)
    = \left\langle \int_{\dom} f(x)\,k_x\,\mu(dx),\, g\right\rangle_H ,
\end{equation*}
and therefore
\begin{equation}
    \label{eq:iotamuadj}
    \iota_\mu^* f = \int_{\dom} k_x f(x)\,\mu(dx).
\end{equation}
The properties of the integral operator $\iota_\mu^*$ are well known; see, for instance,
\cite[Section~4]{christmann2008support}. In particular, $\iota_\mu^*$ is bounded Hilbert--Schmidt operator with $\norm{\iota_\mu^*}_{HS} = \norm{k}_{L^2(\mu)}$ \cite[Theorem~4.27]{christmann2008support}.
To this vein, we also introduce the operator
\begin{equation*}
    B_\mu := \int_{\dom} k_x \otimes k_x \,\mu(dx) : H \to H,
\end{equation*}
where the tensor product is understood in $H$. The operator $B_\mu$ is often called the covariance
(or second-moment) operator in statistical learning theory.
It is compact, positive, self-adjoint, and trace class , and satisfies $B_\mu=\iota_\mu^*\iota_\mu$ with $\tr_H\, B_\mu = \norm{k}_{L^2(\mu)}$ \cite[Theorem~4.27]{christmann2008support}.
Consequently,
\begin{equation}
    \label{eq:normal_operator}
    A_\mu^* A_\mu
    = A^* B_\mu A
    = \int_{\dom} (A^*k_x)\otimes(A^*k_x)\,\mu(dx)
    : U \to U
\end{equation}
is a bounded self-adjoint operator inheriting the properties of $B_\mu$.

\begin{corollary}
The posterior covariance in \eqref{eq:posterior_covariance_form1} satisfies
\begin{equation}
    \label{eq:posterior_covariance_form2}
    C_{post}^\mu = C_f^{\frac 12} (I + C_f^{\frac 12} A_\mu^* A_\mu C_f^{\frac 12})^{-1}C_f^{\frac 12} : U \to U.
\end{equation}
\end{corollary}

\begin{proof}
The operator defined by \eqref{eq:posterior_covariance_form2} is a linear bounded, self-adjoint and trace-class operator in $U$, inheriting these properties from the prior covariance $C_f$ and boundedness of $A_\mu$. To see that the two operators coincide, abbreviate $T_\mu = A_\mu C_f^{\frac 12} : U\to L^2(\mu)$ and observe
\begin{equation*}
    (I + T_\mu^* T_\mu)^{-1} = I - T_\mu^* (T_\mu T_\mu^* + I)^{-1} T_\mu.
\end{equation*}
This proves the claim.
% Let $\{u_j\}_{j=1}^\infty$ be an orthonormal basis of $U$. We have that the Hilbert--Schmidt norm of $C_f^{\frac 12} A_\mu^* A_\mu C_f^{\frac 12}$ on $U$ is bounded, i.e.,
% \begin{eqnarray*}
%     \norm{C_f^{\frac 12} A_\mu^* A_\mu C_f^{\frac 12}}_{\text{HS}(U)}^2
%     & = & \sum_{j=1}^\infty \norm{C_f^{\frac 12} A_\mu^* A_\mu C_f^{\frac 12} u_j}^2 \\
%     & = & \sum_{j=1}^\infty \langle A_\mu^* A_\mu C_f^{\frac 12} u_j, C_f A_\mu^* A_\mu C_f^{\frac 12} u_j\rangle_U \\
%     & \leq & \norm{A_\mu^* A_\mu}_{{\mathcal L}(U)}^2 \norm{C_f}_{{\mathcal L}(U)} \tr_U(C_f) < \infty.
% \end{eqnarray*}
% It follows by \cite[Lemma 3.3]{pinski2015kullback} that the operator $A_\mu^* A_\mu + C_f^{-1}$ can be associated with a quadratic form, which is closed and bounded from below on the Cameron--Martin space of $f$. Furthermore, $(A_\mu^* A_\mu + C_f^{-1})^{-1}$ corresponds to a covariance operator of a Gaussian random variable on $U$.

% It remains to show that the actions coincide.
% Indeed, interpreting $C_g = A_\mu C_f A_\mu^* + I : L^2(\mu) \to L^2(\mu)$, we have
% \begin{eqnarray*}
%     (A_\mu^* A_\mu + C_f^{-1}) C_{post}^\mu
%     & = & A_\mu^* A_\mu  C_f - A_\mu^* A_\mu  C_f A_\mu^* C_g^{-1} A_\mu C_f  + I - A_\mu^* C_g^{-1} A_\mu C_f \\
%     & = & I + A_\mu^*\left[I - A_\mu C_f A_\mu^* C_g^{-1} - C_g^{-1}\right] A_\mu C_f \\
%     & = & I + A_\mu^*\left[I - \left( A_\mu C_f A_\mu^* + I\right) C_g^{-1} + C_g^{-1} - C_g^{-1}\right] A_\mu C_f \\
%     & = & I.
% \end{eqnarray*}
% Reversing the argument yields the claim.
\end{proof}

\begin{example}[An empirical measure]
\label{ex:empirical_measure}
Consider the inference task with design induced by the empirical measure $\mu_B = \frac 1B \sum_{j=1}^B \delta_{x_j} \in {\mathcal P}(\dom)$. It follows that
\begin{equation*}
    C_{post}^{\mu_B} = C_f^{\frac 12}\left(\frac 1B \sum_{j=1}^B (C_f^{\frac 12}A^* k_{x_j}) \otimes (C_f^{\frac 12}A^* k_{x_j}) + I\right)^{-1} C_f^{\frac 12}.
\end{equation*}
This expression coincides with the posterior covariance obtained from the observation model
\begin{equation*}
    g_j = (Af)(x_j) + \sqrt{B} \epsilon_j, \quad, j=1,\ldots, B,
\end{equation*}
where $\epsilon_j\sim {\mathcal N}(0,1)$ i.i.d. Thus, the empirical measure $\mu_B$ corresponds to a design in which independent observations are collected at the sensor locations $x_j$, but each measurement is corrupted by noise whose standard deviation grows as $\sqrt{B}$. In effect, spreading the sampling budget across $B$ locations leads to a higher noise-to-signal ratio at each individual observation. Similarly, starting with $\tilde \mu_B = \sum_{j=1}^B \delta_{x_j}$ leads to independent observations at same locations but with fixed noise-to-signal ratio.
\end{example}

\section{Bayesian OED in the relaxed setting}
\label{sec:BOED}

\subsection{A-optimality criterion}
\label{subsec:Aopt}

Bayesian optimal experimental design seeks to maximize a utility functional over the design space. In this work, the design space is the set of positive measures $\mu$ satisfying $\mu(\dom)=B$, which we henceforth denote by $\Pb$ with the convention that ${\mathcal P}(\dom) = {\mathcal P}^1(\dom)$.
The optimal design is formalized as the maximizer
\begin{align*}
    \mu^* = \argmax_{\mu \in \Pb}\, U(\mu) = \argmax_{\mu \in \Pb}\, \E\, u(f,g^\mu; \mu),
\end{align*}
where the expectation is taken over the joint Bayesian distribution of $f$ and $g^\mu$ constructed from the Gaussian prior on $f$ and the likelihood in \eqref{eq:relaxed_inference}.
In this paper, we focus on the A-optimality criterion that emerges from the choice
\begin{equation}
    u(f,g^\mu; \mu) = - \norm{f - \hat f(g^\mu; \mu)}^2,
\end{equation}
where $\hat f(g^\mu; \mu)$ stands for the mean of the posterior distribution. This leads to the well-known formula for the expected utility, namely, $U$ satisfies
\begin{eqnarray}
    \label{eq:Umu}
    U(\mu) = - \tr_U\, \left[C_{post}(\mu)\right] & = & - \tr_U\, \left[C_f^{\frac 12} (I + C_f^{\frac 12} A_\mu^* A_\mu C_f^{\frac 12})^{-1}C_f^{\frac 12}\right] \nonumber\\
    & = & - \tr_{C_f^{1/2}(U)}\, \left(\int_\dom (C_f^{\frac 12}A^* k_x) \otimes (C_f^{\frac 12} A^* k_x) \mu(dx) + I\right)^{-1}
\end{eqnarray}
by identity \eqref{eq:normal_operator}. In what follows, we write $U_B(\mu)=U(\mu)$ to emphasize the dependence on the domain $\Pb$.
Motivated by the expression in \eqref{eq:Umu}, let us abbreviate
\begin{equation}
    \label{eq:Smu}
    S_\mu = \int_\dom (C_f^{\frac 12}A^* k_x) \otimes (C_f^{\frac 12} A^* k_x) \mu(dx).
\end{equation}
The operator $S_\mu$ can be interpreted as the covariance $B_\mu$ preconditioned by $AC_f^{\frac 12}$.

\begin{proposition}\label{prop:concavity}
The mapping $U_B : \Pb\to \R$ defined by \eqref{eq:Umu} is concave.
\end{proposition}

\begin{proof}
Fix $\mu_0,\mu_1\in\Pb$ and $t\in[0,1]$, and set
\[
\mu_t:=(1-t)\mu_0+t\mu_1,\qquad
M_t:=S_{\mu_t}+I.
\]
By linearity it follows that $M_t=(1-t)M_0+tM_1$.
Each $M_i$ is bounded, self-adjoint, and strictly positive, hence invertible, and the same holds for $M_t$.

We now use the operator convexity (in terms of the Loewner order) of the inverse on the cone of strictly positive bounded self-adjoint operators, i.e., it holds that
\begin{equation}\label{eq:operator-convex-inverse}
\bigl((1-t)M_0+tM_1\bigr)^{-1}\ \leq (1-t)M_0^{-1}+tM_1^{-1}
\end{equation}
for $0\le t\le 1$.
Since the trace is monotone with respect to the Loewner order on positive trace-class
operators \cite[Thm. VI.18 (d)]{ReedSimon1}, applying $\tr_{C_f^{1/2}(U)}(\cdot)$ to \eqref{eq:operator-convex-inverse} yields
\[
\tr_{C_f^{1/2}(U)}(M_t^{-1})
\ \le\ (1-t)\tr_{C_f^{1/2}(U)}(M_0^{-1})+t\tr_{C_f^{1/2}(U)}(M_1^{-1}).
\]
We note that the weight $C_f^{1/2}$ on the trace does not affect the claim as it preserves the Loewner order.
This yields the claim after applying the argument to the negative trace.
\end{proof}

% \textcolor{red}{Are we switching to finite dimensional problems from this section onwards?   If not, we need to adjust things.   Both settings (finite dim / infinite dim) will require different assumptions on the terms.}
 \subsection{Differentiability of the expected utility}

Let us now consider the differentiability of the expected utility $U_B(\mu)$ with respect to the argument $\mu$. Notice that $U_B$ is defined on $\Pb$, which differs slightly from much of the existing literature—particularly on gradient flows—where functionals are typically posed on the space of probability measures ${\mathcal P}(\dom)$. For consistency, we introduce a functional $V_B : {\mathcal P}(\dom) \to \R$ defined by
\begin{equation}
    V_B(\mu) = U_B(B \mu) = - \tr_{C^{1/2}(U)}\left(B \cdot S_\mu + I\right)^{-1}.
\end{equation}
Clearly, optimizing $U_B$ and $V_B$ are equivalent up to the multiplicative constant $\mu\mapsto B\mu$.

Recall that given a functional $F: {\mathcal P}(\dom) \to \R$, we denote by $\frac{\delta F(\mu)}{\delta \mu} : \dom \to \R$, if it exists, the unique (up to additive constants) function such that
\begin{equation}
    \label{eq:first_variation}
    \frac{d}{dt} F(\mu + t \delta \mu)|_{t=0} = \int_\dom \frac{\delta F(\mu)}{\delta \mu}  (x)\delta \mu(dx)
\end{equation}
for every signed measure $\delta \mu$ such that, for some $t_0>0$ and $t\in [0,t_0]$, the measure $\mu + t \delta \mu \in {\mathcal P}(\dom)$. The function $\frac{\delta F}{\delta \mu} (\mu)$ is called \emph{first variation} of the functional $F$ at $\mu$ \cite{gradFlowsOverview}. In order to guarantee differentiability below, we make the following regularity assumption regarding the kernel $k$.

\begin{assumption}
\label{ass:regularity_of_kx}
The feature map $x\mapsto k_x$ belongs to $C^2(\overline \dom, H)$ and
\begin{equation}
    \sup_{x\in\dom} \norm{k_x}_H + \sup_{x\in\dom} \norm{\partial_i k_x}_H + \sup_{x\in\dom} \norm{\partial_i \partial_j k_x}_H < \infty
\end{equation}
for any $i,j = 1,...,n$.
\end{assumption}

\begin{proposition}\label{prop:frechet_derivative}
Suppose Assumption \ref{ass:regularity_of_kx} holds.
The first variation of $V_B$ at $\mu$ satisfies
\begin{equation*}
    \frac{\delta V_B}{\delta\mu}(\mu)(x)
    = B \norm{C_f^{\frac 12} (B S_\mu+I)^{-1} C_f^{\frac 12}A^* k_x}_U^2
\end{equation*}
for all $x\in \dom$.
\end{proposition}

\begin{proof}
Let $\mu$ be fixed and $\delta\mu$ a signed perturbation with $t_0>0$ such that $\mu+t\delta\mu \in {\mathcal P}(\dom)$ for $t\in [0,t_0]$.  
We first observe that $S_{\delta\mu}$ is well-defined by \eqref{eq:Smu} as a linear, self-adjoint and trace-class operator since due to linearity and since $\mu, \mu+t\delta\mu \in {\mathcal P}(\dom)$ one can set 
\begin{equation*}
    S_{\delta \mu} = \frac{S_{\mu + t_0\delta \mu} - S_{\mu}}{t_0}.
\end{equation*}
Now it holds that 
\begin{align*}
V_B(\mu+t\delta\mu)-V_B(\mu)
&= -\tr_{C_f^{1/2}(U)}\,\left[(B(S_\mu+t S_{\delta\mu})+I)^{-1}
             - (BS_\mu + I)^{-1}\right] \\
&= tB\,\tr_{C_f^{1/2}(U)}\!\left[(BS_\mu+I)^{-1} S_{\delta\mu}
        (BS_\mu+I)^{-1}\right] + \mathcal{E}(t),
\end{align*}
using the standard identity
\[
(B(S_\mu+t S_{\delta\mu})+I)^{-1} = (BS_\mu+I)^{-1}-t B (BS_\mu+I)^{-1} S_{\delta\mu}
        (BS_\mu+I)^{-1}+\mathcal{E}(t),
\]
where 
\[
\mathcal{E}(t)
= t^2 B^2 (BS_\mu+I)^{-1} S_{\delta\mu}
        (B(S_\mu+t S_{\delta\mu})+I)^{-1}
        S_{\delta\mu} (BS_\mu+I)^{-1}.
\]
and, consequently, $\norm{\mathcal{E}(t)} = {\mathcal O}(t^2)$.
Using the cyclic property of the trace and the identity $\tr(a\otimes b) = \langle a, b\rangle$ yields
\begin{align*}
\frac{d}{dt} V_B(\mu + t \delta \mu)\big|_{t=0} 
 &= B \tr_{C_f^{1/2}(U)}\!\left[(BS_\mu+I)^{-1} S_{\delta\mu}
        (BS_\mu+I)^{-1}\right] \\
 &= B \int_\dom 
    \tr_{U}\!\left[(BS_\mu+I)^{-1} C_f (BS_\mu+I)^{-1} (C_f^{\frac 12}A^* k_x) \otimes (C_f^{\frac 12} A^* k_x)
        \right] \,\delta\mu(dx) \\
 &= B \int_\dom \langle (B S_\mu+I)^{-1} C_f (B S_\mu+I)^{-1} C_f^{\frac 12}A^* k_x ,\, C_f^{\frac 12} A^* k_x\rangle \,\delta\mu(dx) \\
 &= B \int_\dom \norm{C_f^{\frac 12} (B S_\mu+I)^{-1} C_f^{\frac 12}A^* k_x}_U^2  \,\delta\mu(dx)
\end{align*}
which proves the claim.
\end{proof}

\begin{proposition}
    \label{prop:gradient_of_U}
    Suppose Assumption \ref{ass:regularity_of_kx} holds.
    Then the first variation $\frac{\delta V_B}{\delta \mu}$ is continuously differentiable and its gradient is given componentwise by
    \begin{equation*}
        \left[\partial_i \frac{\delta V_B}{\delta \mu}(\mu)\right](x) = 2 B \left\langle C_f^{\frac 12} (B S_\mu+I)^{-1} C_f^{\frac 12}A^* k_x, C_f^{\frac 12} (B S_\mu+I)^{-1} C_f^{\frac 12}A^* (\partial_i k_x)\right\rangle_U. 
    \end{equation*}
\end{proposition}
\begin{proof}
Since $C_f^{\frac 12} (B S_\mu+I)^{-1} C_f^{\frac 12}A^* : H \to U$ is bounded, our assumption implies that the composition 
\begin{equation*}
    x\mapsto C_f^{\frac 12} (B S_\mu+I)^{-1} C_f^{\frac 12}A^* k_x : \dom \to U
\end{equation*}
is continuously as a map into $U$ and
\begin{equation}
    \partial_i \left[C_f^{\frac 12} (B S_\mu+I)^{-1} C_f^{\frac 12}A^* k_x\right] =
    C_f^{\frac 12}(B S_\mu+I)^{-1} C_f^{\frac 12}A^* (\partial_i k_x).
\end{equation}
Now the result is given by the chain rule.
\end{proof}

\section{Optimization by Wasserstein gradient flow}
\label{sec:Wasserstein_gradient_flow}
 
\subsection{Preliminaries}

We briefly recall the Wasserstein--$2$ geometry on probability measures and the associated notion of gradient flow. Detailed expositions can be found in \cite{gradFlowsOverview,gradientflow_book,villani2021topics}.  Throughout, let $\dom\subset \mathbb{R}^d$ be a closed domain endowed with the Euclidean distance. We denote by $\mathcal{P}_2(\dom) \subset \mathcal{P}(\dom)$ the subset of 
probability measures with finite second moment, i.e., 
\begin{equation*}
\mathcal{P}_2(\dom)
=
\Big\{
\mu\in\mathcal{P}(\dom):
\int_{\dom} |x|^2\,\mu(dx) < \infty
\Big\}.
\end{equation*}

\begin{definition}[Wasserstein $2$--distance]
For $\mu,\nu\in\mathcal{P}_2(\dom)$, the Wasserstein--$2$ distance is defined as
\begin{equation*}
W_2(\mu,\nu)
=
\Bigg(
\inf_{\gamma\in\Gamma(\mu,\nu)}
\int_{\dom\times\dom}
|x-y|^2\,\gamma(dx,dy)
\Bigg)^{1/2},
\end{equation*}
where $\Gamma(\mu,\nu)$ denotes the set of transport plans between $\mu$ and $\nu$, i.e.
\begin{align*}
\Gamma(\mu,\nu)
=
\Big\{
\gamma\in\mathcal{P}(\dom\times\dom):
&\int_{\dom\times\dom}\phi(x)\,\gamma(dx,dy)
=
\int_{\dom}\phi(x)\,\mu(dx),\\
&\int_{\dom\times\dom}\varphi(y)\,\gamma(dx,dy)
=
\int_{\dom}\varphi(y)\,\nu(dy),\\
&\text{for all }\phi,\varphi\in C_b(\dom)
\Big\}.
\end{align*}
\end{definition}
The metric space $(\mathcal{P}_2(\dom),W_2)$ possesses a formal Riemannian structure in which tangent vectors at $\mu$ are identified with velocity fields $v\in L^2(\mu;\mathbb{R}^d)$ through the continuity equation
\begin{equation*}
\partial_t \mu_t + \nabla_x\cdot(\mu_t v_t)=0 .
\end{equation*}
Now let $F : \mathcal{P}_2(\dom) \to \R$ be a functional and recall the definition of the first variation in \eqref{eq:first_variation}.

\begin{definition}{(Wasserstein gradient)} 
Assume that the first variation $\frac{\delta F(\mu)}{\delta \mu} :\dom \to \R$ is continuously differentiable.
The Wasserstein gradient of $F$ at $\mu$ is 
\begin{equation*}
\nabla_{W_2} F(\mu) =  \nabla_x\frac{\delta F(\mu)}{\delta \mu}.
\end{equation*}
\end{definition}
The associated Wasserstein gradient \emph{ascending} flow is characterized by the evolution equation
% \begin{equation*}
% \partial_t \mu_t = -\nabla_{W_2}F(\mu_t) = \nabla_x\cdot \left(\mu_t\nabla_x\frac{\delta F (\mu_t)}{\delta \mu_t}\right),
% \end{equation*}
% which formally decreases $F$ along trajectories, i.e., $\frac{d}{dt} F(\mu_t) \le 0$.  Conversely, as our objective is to \emph{maximize} a functional, one considers the Wasserstein gradient \emph{ascent} flow
\begin{equation}
\label{eq:ascent}
\partial_t \mu_t = -\nabla\cdot\left(\mu_t \nabla_{W_2}(F)(\mu_t)\right) = -\nabla_x\cdot\left(\mu_t\nabla_x\frac{\delta F(\mu_t)}{\delta \mu_t}\right).
\end{equation}
This evolution equation can be interpreted as transporting mass along the velocity field
\begin{equation*}
    v_t(x) = \nabla_x\frac{\delta F(\mu_t)}{\delta \mu_t}(x).
\end{equation*}
The corresponding characteristic curves, i.e., the trajectories followed by individual particles under this transport, satisfy the ordinary differential equation
 \begin{equation} \label{grad_flow_eq}
     \frac{d x(t)}{dt} = v_t(x(t)) = \nabla_{x}\frac{\delta F(\mu_t)}{\delta \mu_t}(x(t)).
 \end{equation}
For numerical computations, the probability measure $\mu_t$ is unknown and is approximated by an empirical distribution of finitely many interacting particles
$\mu_t^N = \frac{1}{N}\sum_{i=1}^N \delta_{x_i(t)}$,
where the particle positions $x_i(t)$ evolve according to the characteristic dynamics above; see e.g. \cite{2ndBatchOptArticle}.  In particular, for the optimization problem
\begin{equation*}
\mu^* = \arg\max_{\mu\in\mathcal{P}_2(\dom)} F(\mu),
\end{equation*}
the measure $\mu_t$ is evolved along the above ascent flow, which transports mass in the direction of increasing first variation $\frac{\delta F}{\delta\mu}$ and converges, under suitable concavity and regularity assumptions, to maximizers of $F$.   

The following theorem establishes that the Wasserstein derivative of $V_B$ exists.

\begin{theorem}[Sufficient conditions for Wasserstein differentiability]\label{thm:wass_diff}
Let $\dom\subset\mathbb R^d$ be a bounded, connected domain with $C^1$ boundary.
Assume the setting of Section~\ref{sec:relaxing_BIP} and, in particular, that $V_B:\mathcal P_2(\dom)\to\mathbb R$ admits the
first variation $\phi_\mu=\frac{\delta V_B(\mu)}{\delta\mu}$ given in Proposition~\ref{prop:frechet_derivative}.
Moreover, suppose Assumption~\ref{ass:regularity_of_kx} holds.
%and, in particular, that $V_B:\mathcal P_2(X)\to\mathbb R$ admits the
%first variation $\phi_\mu=\delta U/\delta\mu$ given in Proposition~3.2 and that
%$x\mapsto \phi_\mu(x)$ is $C^1$ with gradient given by Proposition~3.3.  Assume moreover that:
% \begin{enumerate}
% \item[(A1)] (\emph{Bounded operators}) $A:U\to H$ is bounded and $C_f^{1/2}:U\to U$ is bounded.
% \item[(A2)] (\emph{Kernel regularity}) The feature map $x\mapsto k_x\in H$ is $C^2$ on $\overline X$ and
% \[
% \sup_{x\in X}\|k_x\|_H+\sup_{x\in X}\|\nabla k_x\|_H+\sup_{x\in X}\|\nabla^2 k_x\|_H < \infty.
% \]
% \end{enumerate}
% Then:
% \begin{enumerate}
% \item[(i)] For every $\mu\in\mathcal P_2(\dom)$, the Wasserstein gradient of $V_B$ exists and is given by
% $\nabla\phi_\mu\in L^2(\mu;\mathbb R^n)$.
% \item[(ii)] 
Let $(\mu_t)_{t\in[0,T]}$ be any $W_2$-absolutely continuous curve with
velocity field $v_t\in L^2(\mu_t)$ solving the continuity equation with no-flux boundary condition,
\[
\partial_t\mu_t+\nabla\cdot(\mu_t v_t)=0 \quad \text{in }\dom,
\qquad (\mu_t v_t)\cdot n=0 \quad \text{on }\partial \dom.
\]
Then $t\mapsto V_B(\mu_t)$ is absolutely continuous and, for a.e.\ $t\in[0,T]$,
\[
\frac{d}{dt}V_B(\mu_t)=\int_\dom \nabla\phi_{\mu_t}(x)\cdot v_t(x)\,\mu_t(dx).
\]
% \end{enumerate}
\end{theorem}
\begin{proof}
For convenience, we abbreviate $g(x):=C_f^{1/2}A^*k_x\in U$ and recall that $M_g := 
\sup_{x\in\dom} \norm{g(x)}_U < \infty$.
We write the first variation from Proposition~\ref{prop:frechet_derivative} in the form
\begin{equation*}
\phi_\mu(x)=B\|C_f^{1/2}T_\mu\, g(x)\|_U^2,
\qquad
T_\mu := (BS_\mu+I)^{-1},
\end{equation*}
where $S_\mu$ is given by \eqref{eq:Smu} and, with notation above, satisfies $S_\mu = \int_X g(x)\otimes g(x)\,\mu(dx)$.

Since the mapping $x\mapsto k_x$ is $C^1$ with bounded derivative and the domain $\dom$ is bounded, the function $x\mapsto g(x)$ is Lipschitz continuous; we denote its Lipschitz constant by $L_g$. Moreover, Proposition~\ref{prop:gradient_of_U} gives a formula for $\nabla\phi_\mu$ in terms of
$g$ and $\nabla g(x):=C_f^{1/2}A^*(\nabla k_x)$; by Assumption \ref{ass:regularity_of_kx} we also have $\sup_x\|\nabla g(x)\|<\infty$ and
$\nabla g$ is Lipschitz.

Let $\mu,\nu\in\mathcal P_2(X)$ and let $\pi\in\Pi(\mu,\nu)$ be any coupling. Then
\[
S_\mu-S_\nu
=
\int_{\dom\times\dom}\big(g(x)\otimes g(x)-g(y)\otimes g(y)\big)\,\pi(dx,\,dy).
\]
Using $\|u\otimes u-v\otimes v\|\le (\|u\|+\|v\|)\,\|u-v\|$ for rank-one operators, we obtain
\[
\|S_\mu-S_\nu\|
\le
\int_{\dom\times\dom} (\|g(x)\|+\|g(y)\|)\,\|g(x)-g(y)\|\,\pi(dx,\,dy)
\le
2M_g\,L_g \int_{\dom\times\dom} \|x-y\|\,\pi(dx,\,dy).
\]
By Cauchy--Schwarz,
\[
\int_{\dom\times\dom} \|x-y\|\,\pi(dx,\,dy) \le \left(\int_{\dom\times\dom} \|x-y\|^2\,\pi(dx,\,dy)\right)^{1/2}.
\]
Taking the infimum over couplings yields
\begin{equation}\label{eq:S_lip}
\|S_\mu-S_\nu\|\le 2M_gL_g\, W_2(\mu,\nu).
\end{equation}
Since $S_\mu\ge 0$ we have $\|T_\mu\|\le 1$ for all $\mu$. Using the resolvent identity
\[
T_\mu-T_\nu = (BS_\mu+I)^{-1}B(S_\nu-S_\mu)(BS_\nu+I)^{-1},
\]
we obtain
\[
\|T_\mu-T_\nu\|\le \|T_\mu\|\,B\|S_\mu-S_\nu\|\,\|T_\nu\|\le B\|S_\mu-S_\nu\|
\le 2BM_g L_g\, W_2(\mu,\nu).
\]
Hence $\mu\mapsto T_\mu$ is $W_2$-Lipschitz in operator norm.

From Proposition~\ref{prop:gradient_of_U}, $\nabla\phi_\mu(x)$ is given by
$$
\nabla \phi_{\mu}(x) = 2B\langle C_f^{1/2}T_\mu g(x), C_f^{1/2}T_\mu \nabla g(x) \rangle_U
$$
and hence depends on $(T_\mu, g(x), \nabla g(x))$ polynomially of degree at most two. Since $C_f^{1/2}$ is a bounded linear operator, there exists $M_{C_f} = \lVert C_f^{1/2}\rVert < \infty$.
Using the bounds $\|T_\mu\|\le 1$, $\sup_x\|g(x)\|\le M_g$, $\sup_x\|\nabla g(x)\|\le M_{\nabla g}$ and the
Lipschitz estimates for $T_\mu$ and for $g,\nabla g$, one obtains the following:
there exists $C<\infty$ such that for all $\mu,\nu\in\mathcal P_2(\dom)$ and all $x,y\in \dom$,
\begin{equation}\label{eq:grad_phi_pointwise}
\|\nabla\phi_\mu(x)-\nabla\phi_\nu(y)\|
\le
C\big(\|x-y\| + W_2(\mu,\nu)\big).
\end{equation}
  Let $\pi\in\Pi_{\mathrm{opt}}(\mu,\nu)$ be any optimal coupling, then  $\int \|x-y\|^2\,d\pi = W_2(\mu,\nu)^2$, hence
\begin{equation}\label{eq:plan_lip}
\int_{X\times X}\bigl\|\nabla \phi_\mu(x)-\nabla \phi_\nu(y)\bigr\|^2\,\pi(dx,dy)
\le 4C^2\,W_2(\mu,\nu)^2.
\end{equation}
Therefore the hypotheses of \cite[Theorem~8.3.1]{gradientflow_book} apply to $V_B$.
Let $(\mu_t)_{t\in[0,T]}$ be a $W_2$--absolutely continuous curve with velocity field
$v_t\in L^2(\mu_t)$ solving the continuity equation
$\partial_t\mu_t + \nabla\cdot(\mu_t v_t)=0$ in the sense of distributions
(with the no-flux boundary condition imposed).
Then $t\mapsto V_B(\mu_t)$ is absolutely continuous and for a.e.\ $t\in[0,T]$,
\[
\frac{d}{dt}V_B(\mu_t)=\int_X \nabla \phi_{\mu_t}(x)\cdot v_t(x)\,\mu_t(dx).
\]
\end{proof}

% \subsection{Particle Gradient Flow}
% Now the characteristics in Eq. \ref{grad_flow_eq} forms the particle gradient flow for the OED objective function in Eq. \ref{eq:Umu_batch}. Here $ \frac{\delta U_B(\mu) }{\delta \mu}$
%  denotes the first variation of the functional $U$. The first variations are calculated using the Fréchet derivative for $U$. The use of Fréchet derivative instead of Gateaux derivative is justified as if the Fréchet derivative of $U$ exists, the Gateaux derivative of $U$ exists, and these coincide.\\

% \subsubsection{Gradients for the Particle Gradient Flow}
% Following \cite{2ndBatchOptArticle}, we establish the gradients of the utility function $U(\mu)$ for the Wasserstein gradient flow. First, consider the utility with one measure $\mu$.

\subsection{First Order Optimality Condition}
Having established conditions for Wasserstein differentiability, we now derive first order optimality conditions for (local) maximisers of $V_B$ in $\mathcal{P}_2$, following a similar  approach to \cite{2ndBatchOptArticle} and \cite{lanzetti2025first}. The following lemma establishes necessary and sufficient conditions for stationarity.

\begin{lemma}\label{thm:stationary_point}
Let $\dom\subset\mathbb R^n$ be a bounded, connected domain with $C^1$ boundary. Assume the hypotheses of Theorem~\ref{thm:wass_diff}. Let $(\mu_t)_{t\in[0,T]}$ be a (distributional) solution of the Wasserstein gradient ascent flow \eqref{eq:ascent}.  Then, the map $t\mapsto V_B(\mu_t)$ is non-decreasing and, for a.e.\ $t\in[0,T]$.  A measure $\mu^*\in\mathcal P_2(\dom)$ is a stationary solution of
\eqref{eq:ascent} if and only if
\begin{equation}\label{eq:stationary_weak}
\int_X \nabla\psi(x)\cdot \nabla\phi_{\mu^*}(x)\,\mu^*(dx)=0
\qquad \forall\,\psi\in C^1(\overline \dom).
\end{equation}
In particular, if $\mu^*(dx)=\rho^*(x)\,dx$ with $\rho^*\in L^1(\dom)$ and $\rho^*(x)\ge c>0$ a.e.\ on $\dom$,
then stationarity implies $\nabla\phi_{\mu^*}(x)=0$ for a.e.\ $x\in \dom$ (and hence also $\mu^*$-a.e.).
\end{lemma}

\begin{proof}
% Let $\mu_t$ be the Wasserstein gradient flow for $U$, then
% $$
% \frac{d}{dt} U(\mu_t) = -\int \frac{\delta U(\mu_t)}{\delta \mu}(x) \nabla_{x}\cdot\left(\mu_t(x)\nabla_{x} \frac{\delta U(\mu_t)}{\delta \mu} (x)\right)dx.
% $$
% Integrating by parts, we obtain
% $$
% \frac{d}{dt} U(\mu_t) = \int \left|\nabla_{x} \frac{\delta U(\mu_t)}{\delta \mu} \right|^2 \mu_t(dx) \geq 0.
% $$
% In particular, if $\mu_t = \mu^*$ satisfies condition \eqref{eq:first_order_optimality}, then $\mu^*$ is necessarily a stationary point of the gradient flow.
Write $\phi_t:=\phi_{\mu_t}$. Formally differentiating $V_B(\mu_t)$ along the flow and using
\eqref{eq:ascent} gives
\[
\frac{d}{dt}V_B(\mu_t)
=
-\int_\dom \phi_t(x)\,\nabla\cdot\bigl(\mu_t \nabla \phi_t\bigr)(dx).
\]
Integrating by parts and using the boundary condition we obtain
\[
\frac{d}{dt}V_B(\mu_t)
=
\int_\dom \nabla \phi_t(x)\cdot \nabla \phi_t(x)\,\mu_t(dx)
=
\int_\dom \|\nabla \phi_t(x)\|^2\,\mu_t(dx)\ge 0,
\]
which establishes the monotonicity of $t\mapsto V_B(\mu_t)$.

Stationarity of \eqref{eq:ascent} means $\partial_t\mu_t=0$, i.e.\
$\nabla\cdot(\mu^*\nabla\phi_{\mu^*})=0$ with no-flux. In distributional form this is equivalent to
\eqref{eq:stationary_weak}. For the final claim, assume $\mu^*(dx)=\rho^*(x)\,dx$ with $\rho^*\ge c>0$ a.e.\ and $\phi_{\mu^*}\in H^1(\dom)$.
Then \eqref{eq:stationary_weak} implies that
\[
\int_\dom \rho^*(x)\,\nabla\phi_{\mu^*}(x)\cdot \nabla\psi(x)\,dx=0
\qquad \forall \psi\in H^1(\dom).
\]
Taking $\psi=\phi_{\mu^*}$ yields
\[
\int_\dom \rho^*(x)\,\|\nabla\phi_{\mu^*}(x)\|^2\,dx = 0,
\]
and since $\rho^*(x)\ge c>0$ a.e.\ it follows that $\nabla\phi_{\mu^*}(x)=0$ for a.e.\ $x\in \dom$.
\end{proof}

\begin{remark}
Note that stationary points within the Wasserstein gradient flow do not necessarily correspond to maximisers of $V_B$.   While the functional $V_B$ is concave in the space of signed measures, it is not geodesically / displacement-convex on $\mathcal{P}_2(\dom)$, and so we cannot rely on the typical arguments that are common for gradient flows of geodesically convex energies.
\end{remark}

 We now state the first order optimality condition for local maximisers or $V_B$ in $\mathcal{P}_2(\dom)$.
\begin{proposition}\label{prop:kkt_UB}
Assume the setting and hypotheses of Propositions~3.1--3.3. Then $\mu^*\in\mathcal P(\dom)$ is a maximiser of $V_B$ if and only if there exists a constant $c\in\mathbb R$ such that
the shifted first variation
\[
g(x):=\phi_{\mu^*}(x)-c
\]
satisfies $g(x)\le 0$ for all $x\in \dom$ and $g(x)=0$ $\mu^*$-a.e.
\end{proposition}
\begin{proof}
 Let $\mu^*$ be a maximiser of $V_B$ and fix any $\nu\in\mathcal P(\dom)$.
Define $\mu_t=(1-t)\mu^*+t\nu$. By concavity of $V_B$, the map
$t\mapsto V_B(\mu_t)$ is concave on $[0,1]$, hence its right derivative at $t=0$ is non-positive:
\[
0\ge \frac{d}{dt}V_B(\mu_t)\Big|_{t=0^+}.
\]
By the first variation formula this derivative equals
\[
\frac{d}{dt}V_B(\mu_t)\Big|_{t=0^+}
=\int_\dom \phi_{\mu^*}(x)\,(\nu-\mu^*)(dx)
=\int_\dom \phi_{\mu^*}(x)\,\nu(dx)-\int_\dom \phi_{\mu^*}(x)\,\mu^*(dx).
\]
Let $c:=\int \phi_{\mu^*}d\mu^*$.
Choosing $\nu=\delta_x$ shows $\phi_{\mu^*}(x)\le \int \phi_{\mu^*}\,d\mu^*$ for all $x$, hence in particular
$\phi_{\mu^*}(x)\le c$ for all $x\in\dom$.
Set $g(x)=\phi_{\mu^*}(x)-c\le 0$ on $\dom$. Since $\mu^*$ is supported on $\dom$,
\[
\int_\dom g(x)\,\mu^*(dx)=\int_\dom \phi_{\mu^*}(x)\,\mu^*(dx)-c \le 0.
\]
On the other hand, by definition of $c$ and $g\le 0$, we also have
\[
0\le \int_\dom (c-\phi_{\mu^*}(x))\,\mu^*(dx)= -\int_\dom g(x)\,\mu^*(dx),
\]
so $\int_\dom g\,d\mu^*=0$. Since $g\le 0$ everywhere, this forces $g=0$ $\mu^*$- a.e.
\\\\
Conversely, suppose there exists $c$ such that $g=\phi_{\mu^*}-c\le 0$ on $\dom$ and
$g=0$ on $\mathrm{supp}(\mu^*)$. Then for any $\nu\in\mathcal P(\dom)$
\[
\int_\dom \phi_{\mu^*}\,(\nu-\mu^*)(dx)
=\int_\dom g(x)\,\nu(dx)-\int_\dom g(x)\,\mu^*(dx)
\le 0,
\]
since $g\le 0$ and $g=0$ on $\mathrm{supp}(\mu^*)$. By concavity of $V_B$,
\[
V_B(\nu)\le V_B(\mu^*) + \frac{d}{dt}V_B((1-t)\mu^*+t\nu)\Big|_{t=0^+}
= V_B(\mu^*) + \int_\dom \phi_{\mu^*}\,(\nu-\mu^*)(dx)
\le V_B(\mu^*).
\]
Thus $\mu^*$ is a maximiser of $V_B$.
\end{proof}
% \begin{proposition}
% A measure $\mu \in \mathcal{P}(X)$ is a maximiser of $U$ if and only if there exists a constant $c\in \mathbb{R}$ such that the shifted first variation
% $$
% g(x) = \frac{\delta U(\mu^*)}{\delta \mu}(x) - c,
% $$
% satisfies $g(x) \leq 0$ for all $x \in X$ and $g(x) = 0$ for  $x \mu^*$-a.e.
% \end{proposition}
% \textcolor{red}{TODO: Add a proof or find a reference for this standard result.}

% \begin{proposition}
% A measure $\mu \in \mathcal{P}(\mathcal{X})$ such that $U(\mu) < \infty$ maximised $U$ if and only if $\frac{\delta U(\mu)}{\delta \mu} \leq 0$ and $\frac{\delta U(\mu)}{\delta \mu}(x) = 0$ for all $x \in \mbox{supp}(\mu)$.
% \end{proposition}
Generally, there are strong differences between   stationarity and global optimality.    However, we do know that the global maximiser of $U$ is itself a stationary point of $U$ in $\mathcal{P}_2$.  

\begin{proposition}\label{prop:maximiser_stationary_UB}
Let $\dom\subset\mathbb R^n$ be a bounded, connected domain with $C^1$ boundary and assume the hypotheses of Theorem~\ref{thm:wass_diff}.  Let $\mu^*\in\mathcal P_2(\dom)$
be a \emph{global maximiser} of $V_B$ over $\mathcal P_2(\dom)$. Then $\mu^*$ is a stationary solution of \eqref{eq:ascent} in the distributional sense.
\end{proposition}
\begin{proof}
Fix $\psi\in C^1(\overline \dom)$ with $\partial_n\psi=0$ on $\partial \dom$ and consider the (smooth) velocity field
\[
v := \nabla\psi,
\qquad\text{so that}\qquad v\cdot n = \partial_n\psi = 0\ \text{ on }\partial \dom.
\]
Let $(\mu_t)_{t\in[0,T]}$ be any $W_2$-absolutely continuous curve solving the continuity equation
\[
\partial_t\mu_t+\nabla\cdot(\mu_t v)=0 \quad\text{in }\dom,
\qquad (\mu_t v)\cdot n=0 \quad\text{on }\partial \dom,
\]
with initial condition $\mu_0=\mu^*$; such a curve can be realised by the flow map of $v$ for small times.
By global optimality of $\mu^*$ we have $V_B(\mu_t)\le V_B(\mu^*)$ for all sufficiently small $t\ge 0$, hence
the right derivative at $t=0$ satisfies
\[
\frac{d}{dt}V_B(\mu_t)\Big|_{t=0^+}\le 0.
\]
Applying the chain rule from Theorem~\ref{thm:wass_diff}(ii) at $t=0$ gives
\[
\frac{d}{dt}V_B(\mu_t)\Big|_{t=0^+}
=
\int_\dom \nabla\phi_{\mu^*}(x)\cdot v(x)\,\mu^*(dx)
=
\int_\dom \nabla\phi_{\mu^*}(x)\cdot \nabla\psi(x)\,\mu^*(dx),
\]
and therefore
\begin{equation}\label{eq:ineq_plus}
\int_\dom \nabla\phi_{\mu^*}(x)\cdot \nabla\psi(x)\,\mu^*(dx)\le 0
\qquad \forall\,\psi\in C^1(\overline \dom)\text{ with }\partial_n\psi=0.
\end{equation}
Applying the same argument to $-\psi$ yields the reverse inequality, hence equality in \eqref{eq:ineq_plus}:
\[
\int_\dom \nabla\phi_{\mu^*}(x)\cdot \nabla\psi(x)\,\mu^*(dx)=0
\qquad \forall\,\psi\in C^1(\overline \dom)\text{ with }\partial_n\psi=0.
\]
This is  the distributional stationarity condition
$\nabla\cdot(\mu^*\nabla\phi_{\mu^*})=0$ in $\dom$ with no-flux boundary condition, as required.
\end{proof}

\begin{proposition} All local maximisers of $V_B$ are global maximisers in $\mathcal{P}_2(\dom)$.
\end{proposition}
\begin{proof}
Suppose that $\mu_0$ is a local maximum of $V_B$, so that, for some $\delta > 0$, for all $\mu \in \mathcal{P}(\mathcal{X})$ such that $W_2(\mu, \mu_0) \leq \delta$ we have $V_B(\mu_0) > V_B(\mu)$.  Suppose also that $V_B(\mu_0) < V_B(\mu^*)$.  
\\\\
Consider the mixture $\mu_t = (1-t)\mu_0 + t \mu^*$, for $t\in [0,1]$, then by concavity
$$
V_B(\mu_t) \geq (1-t)V_B(\mu_0) + t V_B(\mu^*) > V_B(\mu_0).
$$
Noting that 
$$
W_2(\mu_t, \mu_0) \leq \sqrt{t} W_2(\mu^*, \mu_0),
$$
then choosing $t$ sufficiently small yields a contradiction.  It follows that $V_B(\mu_0) = V_B(\mu^*)$.
\end{proof}

While concavity implies that all local maximisers of $V_B$ are themselves global optimisers, this is not initself sufficient to establish a unique global maximiser.   To see this, we consider the following toy example.  For simplicity we work over the unit torus, to ignore effects of the boundary, however we note that similar counterexamples can be established in the more general setting.

\begin{example} Consider $\dom := \mathbb T = \mathbb R / \mathbb Z$ and let $k(x,y) := \cos\big(2\pi(x-y)\big)$ with $x,y\in\mathbb T$.
This kernel is positive semi-definite, translation invariant and admits the finite--dimensional feature representation
\[
k(x,y) = \phi(x)^\top \phi(y),
\qquad
\phi(x) :=
\begin{pmatrix}
\cos(2\pi x)\\
\sin(2\pi x)
\end{pmatrix}
\in \mathbb R^2 .
\]
The associated RKHS $H$ is two--dimensional and consists of functions of the form
\[
f(x)=\theta^\top \phi(x), \qquad \theta\in\mathbb R^2 .
\]
We consider the inverse problem, where $A(x)f = f(x)$, i.e.
\[
y = A(x)f + \varepsilon = f(x) + \varepsilon
= \theta^\top \phi(x) + \varepsilon,
\qquad \varepsilon\sim\mathcal N(0,1),
\]
with Gaussian prior
\[
\theta \sim \mathcal N(0,\sigma^2 I_2).
\]
For a given probability measure $\mu\in {\mathcal P}(\mathbb T)$, the covariance operator $B_\mu$ has the matrix representation
\begin{equation}
\label{eq:example_Bmu}
B_\mu = \int_{\mathbb T} \phi(x)\phi(x)^\top \,\mu(dx)
= \frac 12 \int_{\mathbb T} \begin{pmatrix}
    1 + \cos(4\pi x) & \sin (4\pi x) \\
    \sin(4 \pi x) & 1 - \cos(4 \pi x)
\end{pmatrix}
\,\mu(dx)
\in \mathbb R^{2\times 2}.
\end{equation}
Since $\|\phi(x)\|^2=1$ for all $x$, every feasible $B_\mu$ satisfies $\mbox{Tr}(B_\mu)=1$, and with the posterior covariance of $\theta$ given by $C_{\mathrm{post}}(\mu)
= \big(\sigma^{-2} I_2 + B_\mu\big)^{-1}$,
the expected utility with $B=1$ is
\[
U_1(\mu)
= -\tr\big(C_{\mathrm{post}}(\mu)\big)
= -\tr\big(\sigma^{-2} I_2 + B_\mu\big)^{-1}.
\]
Let $\lambda_1,\lambda_2$ be the eigenvalues of $B_\mu$.  
Since $\lambda_1+\lambda_2=1$, convexity of $t\mapsto (\sigma^{-2}+t)^{-1}$ implies
\[
U(\mu) \le
-2(\sigma^{-2}+1/2)^{-1},
\]
with equality if and only if
\begin{equation}
    \label{eq:example_condition}
B_\mu = \tfrac12 I_2.
\end{equation}
We note that with repeated eigenvalue, the symmetry of $B_\mu$ implies the latter identity. Thus the covariance operator is unique, but observe that the maximising measure need not be. In fact, there exist infinitely many maximising measures $\mu$, since the condition \eqref{eq:example_condition} is equivalent to requiring that
\begin{equation}
    \int_{\mathbb T} \cos(4 \pi x) \mu(dx) = 0\quad \text{and} \quad 
    \int_{\mathbb T} \sin(4 \pi x) \mu(dx) = 0.
\end{equation}
Consequently, characterizing when the maximising measure is unique depends intrinsically on the observational model and lies beyond the scope of this work.

% Consider non-uniqueness through the lens of a empirical measure $\mu_N = \frac1N\sum_{i=1}^N \delta_{x_i}$ for which
% the covariance is given by
% \[
% B_{\mu_N} = \frac1N\sum_{i=1}^N \phi(x_i)\phi(x_i)^\top .
% \]
% One can, for example, define the four--point empirical measure
% \[
% \mu^{(1)} := \frac14\Big(
% \delta_0 + \delta_{1/4} + \delta_{1/2} + \delta_{3/4}
% \Big).
% \]
% The corresponding feature vectors are
% \[
% \phi(0)=(1,0),\quad
% \phi(1/4)=(0,1),\quad
% \phi(1/2)=(-1,0),\quad
% \phi(3/4)=(0,-1),
% \]
% and a direct computation yields
% \[
% B_{\mu^{(1)}} = \frac12 I_2.
% \]

% Now define a shifted design
% \[
% \mu^{(2)} := \frac14\sum_{j=0}^3 \delta_{\,1/8 + j/4}.
% \]
% Since $\phi(x)\phi(x)^\top$ rotates under translations and the four--point design is symmetric,
% \[
% B_{\mu^{(2)}} = \frac12 I_2
% \]
% as well. Hence
% \[
% U(\mu^{(1)}) = U(\mu^{(2)}) = \max_{\mu\in P(\mathbb T)} U(\mu),
% \qquad
% \mu^{(1)} \neq \mu^{(2)}.
% \]
\end{example}

\section{Particle Approximation and Regularisation}
\label{sec:regularization}

In previous sections, we have established a Wasserstein gradient flow approach to optimize $U_B$ or, more precisely, the normalized utility $V_B$.
The method for optimizing $U_B$ is summarized in Algorithm~\ref{alg:1}, where we utilize the forward Euler method similar to \cite{2ndBatchOptArticle}.

\begin{algorithm}    
\caption{Particle Gradient Flow for Batch-based OED}\label{alg:1}
    \algorithmicrequire \text{} Number of particles $N$, number of iterations $T$, time step $dt$, regularization parameters $\alpha$ and $\beta$, batch size $B$, set of initial particles $x_1^0,...,x_N^0 \in \dom$, prior covariance $C_f$. \\
    \algorithmicensure \text{} Finite positive measure $\nu \in \Pb.$
\begin{algorithmic}[1]
     \For{$t = 1,...,T$}
     \For{$i = 1,..., N$}
        \State $x_i^t \leftarrow x_i^{t-1} + {\rm dt}\nabla_{x}\frac{\delta V_B(\mu_N^{t-1})}{\delta \mu_N^{t-1}}(x_i^{t-1})$
     \EndFor
     \State 
          $\mu_{N}^t \leftarrow \frac{1}{N}\sum_{i = 1}^N\delta^t_{x_i}$
     \EndFor
     \State Set $\nu = B \cdot \mu_N^T$.
\end{algorithmic}
\end{algorithm}

While Algorithm~\ref{alg:1} has theoretical guarantees of convergence, there can be practical constraints in the sensor placement task that need to be considered. For a batch size $B>1$, such constraints can involve requiring exactly $B$ locations in the output. In the vanilla implementation of the method, this is currently not guaranteed and the support of the maximizing measure can vary. Secondly, the sensors may have a physical constraints related to their relative locations, e.g., we cannot place two sensors too close to each other. Such considerations require us to separate the mass associated to each design variable in $\mu \in \Pb$ and enforce preferable solutions through regularization.

In order to enforce these ideas, we lift the problem back to the tensor space $\dom^B = \dom \times \ldots \times \dom$ and consider optimization of a product measure 
$\bmu = \otimes_{j=1}^B \mu_j$ over $\otimes_{j=1}^B {\mathcal P}(\dom)$. We formulate a utility function $\bUb$ for $\bmu$ by setting
\begin{equation}
    \label{eq:UBbmu}
    \bUb(\bmu) := U_B\left(\sum_{j=1}^B \mu_j\right) = V_B\left(\frac 1B\sum_{j=1}^B \mu_j\right).
\end{equation}
We immediately observe that $\bUb$ preserves the concavity and regularity properties of $V_B$. In the same vein, we formally have
\begin{eqnarray*}
    \frac{d}{dt} \bUb(\bmu + t \delta \bmu) \big|_{t=0} 
    & = & \frac{d}{dt} V_B\left(\sum_{j=1}^B \mu_j + t \sum_{j=1}^B \delta \mu_j\right) \\
    & = & \sum_{j=1}^B \int_\dom B \cdot \norm{C_f^{\frac 12} (B S_\mu + I)^{-1} C_f^{\frac 12} A^* k_x}_U^2 \delta \mu_j(dx) \\
    & = & \int_{\dom^B} B\sum_{j=1}^B \norm{C_f^{\frac 12} (B S_\mu + I)^{-1} C_f^{\frac 12} A^* k_{x_j}}_U^2 \delta \bmu(d\bx),
\end{eqnarray*}
where we utilize notation $\bx = (x_1, ..., x_B)$. It follows that one has 
\begin{equation*}
    \frac{\delta \bUb(\bmu)}{\delta \bmu}(\bx) = B\sum_{j=1}^B \norm{C_f^{\frac 12} (B S_\mu + I)^{-1} C_f^{\frac 12} A^* k_{x_j}}_U^2
\end{equation*}
and, consequently,
\begin{equation*}
    \partial_{j,i} \frac{\delta \bUb(\bmu)}{\delta \bmu}(\bx)
    = \left[\partial_i \frac{\delta V_B}{\delta \mu}(\mu)\right](x_j) 
\end{equation*}
where $\partial_{j,i}$ stands for the derivative applied to the component $i$ of vector $x_j$.

Next, we introduce two regularization functionals acting on such product measures and derive their role in Wasserstein gradient flows. The first functional promotes concentration within each batch, while the second introduces repulsion between different batches in state space.
The first regularization functional we utilize is the variance
\begin{equation}
\mathcal R_v(\bmu) = \sum_{j=1}^B \E^{\mu_j} |x-\E^{\mu_j} x|^2,
\end{equation}
which naturally promotes the concentration of each individual ensemble in the gradient flow.
The second regularizer is a \emph{repulsive} interaction energy based on kernel embeddings,
which encourages different measures to separate in state space. 
To define it, let $q:\dom\times\dom\to\mathbb R$ be a symmetric, continuously differentiable
positive definite kernel. Now we set
\begin{equation}
\mathcal R_r(\boldsymbol{\mu}) = \sum_{j\neq j'} \iint_{\dom \times \dom} q(x,x')\, \mu_j(dx)\, \mu_{j'}(dx').
\end{equation}
Up to additive self--interaction terms
$\int q(x,x')\,\mu_j(dx)\mu_j(dx')$,
this functional coincides with the negative squared
maximum mean discrepancy between pairs of measures in the reproducing
kernel Hilbert space associated with $q$.
As such, minimizing $-\mathcal R_r$ encourages separation of the
distributions $\mu_j$ in feature space and induces mutual repulsion in
the corresponding Wasserstein dynamics.

Adding the regularization terms to the utility function, we define
\begin{equation*}
    H_{\alpha,\beta}(\bmu) = U_B(\bmu) - \alpha \mathcal{R}_r(\bmu) - \beta \mathcal{R}_v(\bmu),
\end{equation*}
where $\alpha,\beta \geq 0$ stand for the regularization parameters,
and study the gradient flow
\begin{equation*}
     \partial_t \bmu = -\nabla_{W_2}(\bmu) = \nabla_{x}
    \left( \bmu \nabla_{x} \frac{\delta H_{\alpha, \beta}(\bmu)}{\delta \bmu}\right).
\end{equation*}

Let us next establish the first variation and gradients corresponding to the regularization terms.
\begin{proposition}
For $\mu\in\mathcal P_2(\dom)$, we have
\begin{equation}
\frac{\delta \mathcal R_v(\bmu)}{\delta\bmu}(\bx) = \sum_{j=1}^B |x_j-\E^{\mu_j} x|^2
\quad \text{and} \quad
\nabla_{x_i}
\frac{\delta\mathcal R_v(\bmu)}{\delta\bmu}(\bx) = 2 (x_i-\E^{\mu_i} x).
\end{equation}
\end{proposition}

\begin{proof}
Let $\delta\bmu = \otimes_{j=1}^B \delta \mu_j$ be a signed measure such that there exists $t_0>0$ such that $\bmu + t\delta \bmu \in \otimes_{j=1}^B\mathcal P_2(\dom)$. Write $m(\mu) = \int_\dom x \mu(dx)$ for shorthand. We have
\begin{eqnarray*}
    \mathcal R_v(\bmu+t\delta \bmu) & = & 
    \sum_{j=1}^B \int_\dom | x - m(\mu_j) - t m(\delta \mu_j)|^2 (\mu_j(dx) + t \delta \mu_j(dx)) \\
    & = & \mathcal R_v(\bmu)
    + \sum_{j=1}^B \left[t^2 |m(\delta \mu_j)|^2 + t \int_\dom |x-m(\mu_j)|^2 \delta \mu_j(dx) + t^3 |m(\delta \mu_j)|^2\right]
\end{eqnarray*}
and, therefore, 
\[
\frac{d}{dt}
\mathcal R_v(\bmu + t \delta \bmu)\Big|_{t=0}
=
\sum_{j=1}^B \left[\int_\dom |x-m(\mu_j)|^2\,\mu_j(dx)\right]
= 
\int_{\dom^B} \left[\sum_{j=1}^B|x_j-m(\mu_j)|^2\right] \bmu(d\bx).
\]
Now the formula for the gradient follows directly.
\end{proof}

% \begin{proposition} \label{Gradient of variance regularizer}
%     The gradient of $\mathcal{R}_v'$ for the Wasserstein particle gradient flow is \[\nabla_{x}\mathcal{R}_v'= \sum_{j=1}^d \left[ 2\left(x_{j} -\frac{1}{N}\sum_{i=1}^Nx_j^i\right)\left(1-\frac{1}{N}\right)\right].\]
% \end{proposition}
% \textcolor{red}{[Can we discuss this in our Friday call - I could not reproduce the same result.]}
% \textcolor{blue}{indexing fixed}
%     \begin{proof} of Proposition \ref{Gradient of variance regularizer}\\\\
%         The regularizer is \[\mathcal{R}_v = \sum_{j=1}^d \left[ \frac{1}{n} \sum_{i=1}^n(x_j^i -\overline{x}_j)^2\right],\] that is, a sum of empirical variances of the measures corresponding to each design parameter.
%         The first variation of the variance of $\mu$ is $(x -\overline{x})^2$.
%         Now considering this as the empirical variance and taking the gradient with respect to $x$ we have
%         \begin{align*}
%             \nabla_{x} \mathcal{R}_v' &= \nabla_{x} \sum_{j=1}^d (x_{j}-\overline{x}_j)^2 = \nabla_{x} \left[ \sum_{j=1}^d \left(x_{j}-\frac{1}{N} \sum_{i=1}^Nx_j^i\right)^2\right]\\
%             &= \sum_{j=1}^d 2\left( x_{j}-\frac{1}{N}\sum_{i=1}^Nx_j^i \right)
%             \left(1-\frac{1}{N} \right).
%         \end{align*}
%     \end{proof}

\begin{proposition} \label{First variation of repulsive regularizer}
    For $\mu\in\mathcal P_2(\dom)$, we have
    \begin{equation*}
        \frac{\delta {\mathcal R}_r(\bmu)}{\delta \bmu}(\bx) = 2 \sum_{j\neq j'} \iint_{\dom\times \dom} q(x_j,x)\, \mu_{j'}(dx).
    \end{equation*}
    and
    \begin{equation*}
        \nabla_{x_i} \frac{\delta {\mathcal R}_r(\bmu)}{\delta \bmu}(\bx)
        = 2 \sum_{\substack{j=1 \\ j\neq i}}^B\int_\dom \nabla_{x_i}q(x_i,x) \mu_{j}(dx).
    \end{equation*}
\end{proposition}
\begin{proof}
Let $\delta\bmu = \otimes_{j=1}^B \delta \mu_j$ be a signed measure such that there exists $t_0>0$ such that $\bmu + t\delta \bmu \in \otimes_{j=1}^B\mathcal P_2(\dom)$. 
We first observe that
\begin{multline*}
    {\mathcal R}_r(\bmu + t \delta \bmu) - {\mathcal R}_r(\bmu) \\
    = \sum_{j\neq j'} \iint_{\dom\times \dom} q(x,x') 
    \left(t \mu_j(dx) \delta \mu_{j'}(dx') + t \delta \mu_j(dx) \mu_{j'}(dx)
    + t^2 \delta \mu_j(dx) \delta \mu_{j'}(dx')
    \right)
\end{multline*}
and, consequently,
\begin{eqnarray*}
    \frac{d}{dt} {\mathcal R}_r(\bmu + t \delta \bmu)\big|_{t=0}
    & = & 2 \sum_{j\neq j'} \iint_{\dom\times \dom} q(x,x') \delta \mu_j(dx) \mu_j(dx') \\
    & = & 2 \sum_{j=1}^B \int_\dom \left[\sum_{\substack{j'=1 \\ j'\neq j}}^B\int_\dom q(x,x') \mu_{j'}(dx') \right] \delta \mu_j(dx) \\
    & = &  \int_\dom \left[2 \sum_{j=1}^B\sum_{\substack{j'=1 \\ j'\neq j}}^B\int_\dom q(x_j,x') \mu_{j'}(dx') \right] \delta \bmu(d\bx) \\
    & = &  \int_\dom \left[2 \sum_{j\neq j'}\int_\dom q(x_j,x') \mu_{j'}(dx') \right] \delta \bmu(d\bx).
\end{eqnarray*}
This yields the claim.
\end{proof}

% The first variation is calculated by using the Gateaux derivative.
% Consider the regularizer $\mathcal{R}_r$ with two measures $\mu_d, \mu_l$
% \begin{align*}
%     \mathcal{R}_r(\mu_d,\mu_l) = -2\int\int q(x_d,x_l)\mu_k(dx_d)\mu_l(dx_l).
% \end{align*}
%  For $\epsilon > 0 $ and $\mu + \epsilon\nu \in \mathcal{P}(\dom)$
% \begin{align*}
%     \frac{\partial \mathcal{R}(\mu_d +\epsilon \nu,\mu_l )}{\partial \epsilon} \Big|_{\epsilon=0} = \frac{\partial }{\partial \epsilon}\left[
%     -2 \int_\dom \int_\dom q(x_d,x_l) (\mu_d+ \epsilon \nu)(dx_d)\mu_l(dx_l)\right]\Big|_{\epsilon=0} =   -2\int_\dom q(x_d,x_l)\mu_l(dx_l)
% \end{align*}
  
% \begin{proposition}\label{gradient MMD reg}
%     The gradient of $\mathcal{R}_r'$ for the Wasserstein particle gradient flow is the following.
%     \begin{align*}
%         \nabla_{x_d} \mathcal{R}_r(\mu_1,..,\mu_d)' = \sum_{d\neq l}^L \left[-\frac{2}{N} \sum_{i=1}^N \nabla q(x_d,x_l^{(i)}) \right].
%     \end{align*}
    
% \end{proposition}
% \begin{proof} of Proposition \ref{gradient MMD reg}\\\\
%     Consider the empirical measures for $\mathcal{R}_r'$ and we have
%     \begin{align*}
%         \mathcal{R}_r(\mu_{1_n},...,\mu_{d_n})' = \sum_{d\neq l}^L \left[-\frac{2}{M}  \sum_{i=1}^M q(x_d,x_l^{(i)})\right].
%     \end{align*}
%     The gradient of this with respect to $x$ is 
%     \begin{equation*}
%         \sum_{d\neq l}^L \left[-\frac{2}{N} \sum_{i=1}^N \nabla q(x_d,x_l^{(i)}) \right].
%     \end{equation*}
% \end{proof}

Algorithm \ref{alg:2} summarizes the regularized Wasserstein gradient flow.

\begin{algorithm}    
\caption{Regularized Particle Gradient Flow for Batch-based OED}\label{alg:2}
    \algorithmicrequire \text{} Batch size $B$, number of particles $N$ per ensemble, number of iterations $T$, time step $dt$, regularization parameters $\alpha$ and $\beta$, set of initial particles $\bx_1^0,...,\bx_N^0 \in \dom^B$, prior covariance $C_f$. \\
    \algorithmicensure \text{} Probability measure $\nu \in \otimes_{j=1}^B {\mathcal P}(\dom).$
\begin{algorithmic}[1]
     \For{$t = 1,...,T$}
     \For{$i = 1,..., N$}
     \For{$b = 1,...,B$}
        \State $x_{i,b}^t \leftarrow x_{i,b}^{t-1} + {\rm dt}\left[\nabla_{x_b}\frac{\delta \bUb(\bmu_N^{t-1})}{\delta \bmu}(\bx^{t-1}_i)-\alpha \nabla_{x_b}\frac{\delta {\mathcal R}_v(\bmu_N^{t-1})}{\delta \bmu}(\bx^{t-1}_i) - \beta \nabla_{x_b} \frac{\delta {\mathcal R}_r(\bmu_N^{t-1})}{\delta \bmu}(\bx^{t-1}_i)\right]$
     \EndFor
        \State $\bx_i^t = (x_{i,1}^t, ..., x_{i,B}^t)$        
     \EndFor
     \For{$b = 1,...,B$}
        \State $\mu_{N,b}^t \leftarrow \frac{1}{N}\sum_{i = 1}^N\delta^t_{x_{i,b}}$
     \EndFor
     \State $\bmu_N^t \leftarrow \otimes_{b=1}^B \mu_{N, b}^t$
     \EndFor
    \State Set $\nu = \bmu_N^T$.
\end{algorithmic}
\end{algorithm}

\section{Numerical Simulations}
\label{sec:numerical_simulations}
In this section, we demonstrate the use of the algorithms proposed in previous section with numerical examples in the setting of inverse source term problem for the Poisson's equation and the time-harmonic Schr\"odinger equation. Furthermore, we test the sensitivity of the $\mathcal{R}_v$ and $\mathcal{R}_r$ regularizers to show how the selection of the regularization parameters affects.

\subsection{Inverse source problem for one-dimensional Poisson's equation}
\label{subsec:poisson}
Consider the one-dimensional Poisson's equation with homogeneous Dirichlet boundary conditions
\begin{equation}\label{eq:poisson_1D}
\left\{
\begin{aligned}
-\dfrac{\partial^2 u}{\partial x^2}(x) &= f(x), && x \in \dom=[0,1],\\
u(x) &= 0, && x \in \partial\dom .
\end{aligned}
\right.
\end{equation}
Here, the inverse problem is to identify the source term $f$ given noisy point observations of the solution $u$ in $\dom$.
For a source $f\in H^1(0,1)$, it is well-known that the problem has a unique solution 
$u\in H$, where
\[
H := H^3(0,1)\cap H_0^1(0,1), \qquad 
\langle u,v\rangle_H := \int_\dom \frac{\partial^3 u}{\partial x^3}(x)\,\frac{\partial^3 v}{\partial x^3}(x)\,dx .
\]
By Sobolev embedding theorem, the embedding $H\hookrightarrow C^2(\dom)$ is continuous.
Consequently, point evaluation and its first two derivatives are continuous linear functionals on $H$, so $H$ is a reproducing kernel Hilbert space. In particular, the mapping $x\mapsto k_x\in H$ is twice continuously Fréchet differentiable.
We note that $C^2(\dom)$ stands here for functions on $C^2(0,1)$ that are continuously extended to the closed interval.

% The solution $u$ admits the representation
% \begin{equation*}
% u(z)=\int_{\Omega} G(z,s)\,f(s)\,ds \quad \text{with} \quad
% G(z,s)=
% \begin{cases}
% z(1-s), & z\le s,\\
% s(1-z), & z>s.
% \end{cases}
% \end{equation*}
Our objective is to identify optimal configuration of $B$ sensor locations for measuring $u$ simultaneously. In the context of previous analysis, the linear operator $A$ describes the mapping $f \mapsto u : H^1(0,1) \to H$. Each independent point observation is contaminated with additive zero-mean Gaussian noise with variance 0.01.
For the numerical implementation, the operator $A$ is discretized on a grid of $N=100$ equidistant points. Since the Green's function for \eqref{eq:poisson_1D} is classical, the output function $f$ can be evaluated continuously on $\dom$.
In this example, we impose a non-stationary zero-mean Gaussian prior distribution for $f$ based on covariance function
\begin{equation*}
    c(x,z) = \left(1 + 50\left(x-\frac 12\right)^2\right)\left(1 + 50\left(z-\frac 12\right)^2\right)
    \exp\left(-\frac{|x-z|^2}{2\sigma_0^2}\right),
\end{equation*}
where $\sigma_0 = 0.01$.
We simulate Algorithm~\ref{alg:1} using 120 particles initialized from a uniform distribution over $\dom$. For simplicity, we take a batch size $B=2$ and regularize the empirical variance with weight $\alpha=0.008$. Here, no repulsive regularizer was applied, i.e. $\beta=0$. Using 500 time steps of size $4e-3$, we illustrate the results in Figures~\ref{fig:poisson_ut} and \ref{fig:ex1 trajectories and distb}.

First, we display the expected utility for the two design variables together with the two true optimal configurations marked with red dots $(x_1,x_2)\approx(0.1616,0.8384)$ and $(x_1,x_2)\approx(0.8384, 0.1616)$ (note the symmetry). The optimal configuration found via Algorithm 1 is marked with the cross. Second, Figure~\ref{fig:ex1 trajectories and distb} shows the evolution of the particle trajectories as well as the resulting empirical distribution of the measure $\mu$. We observe that the empirical distribution converges to the optimal configuration.

% We consider a Gaussian prior covariance 
% \begin{align*}\label{Eq:gauss_pri}
% &C_f = \rm{diag}(v)\cdot C \cdot \rm{diag}(v), \text{ where}\\
%  &   C = \gamma^2\exp\left(-\frac{|t_i-t_j|^2}{2\sigma_0^2}\right), \text{ and }
%  v = 1 + 50\left( \frac{i}{N}-\left(\frac{1}{2}\right)^2\right), i = 1,...,N.
% \end{align*}
% The domain is discretized with $N=30$ discretization points. The initial set of particles is sampled from a uniform distribution on the unit interval. We consider the optimization problem with $d= 2$ design variables and apply the regularizer $\mathcal{R}_v$. Further parameters used in the simulation are presented in table \ref{tab: parameters examples}.

%Figure \ref{fig:poisson_ut} presents the utility $U_B(x_1,x_2)$  on the discretized domain with respect to the design variables $x_1,x_2$ as well as the optimal measurement locations  $x =(x_1, x_2)$.
%We obtain the following optimal measurement locations $x =(x_1, x_2)$ as presented in the Fig. \ref{fig:poisson_ut}.
\begin{figure}[h!]
    \centering
    \includegraphics[width=0.6\linewidth]{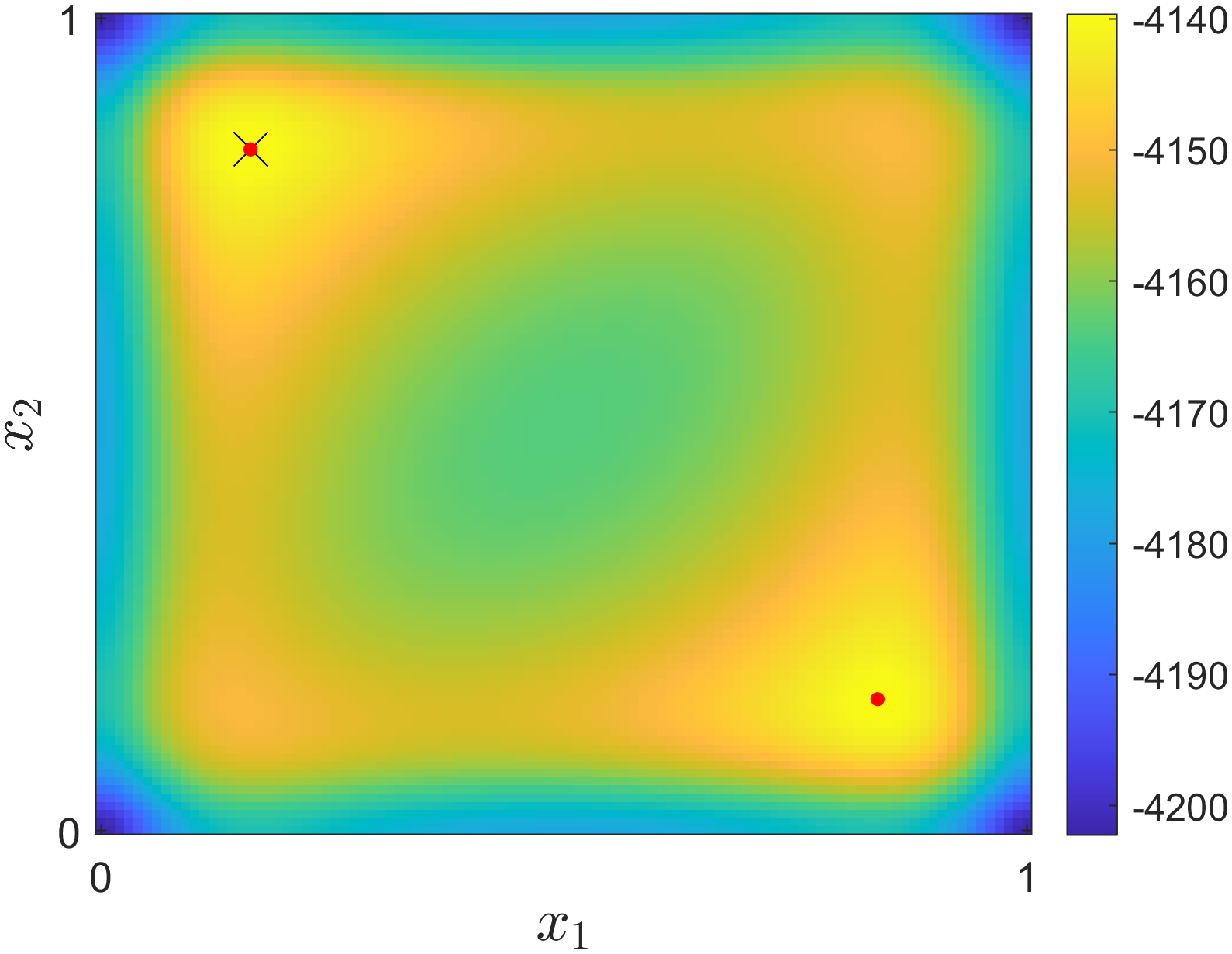}
    \caption{Expected utility w.r.t. to the design variables $x_1$ and $x_2$ in Section \ref{subsec:poisson} is illustrated. The black cross represents the optimal configuration $x \approx (0.1614,0.8386)$ found. The red dots mark the true optimal configurations.}
    \label{fig:poisson_ut}
\end{figure}

%{\color{red}
%\begin{itemize}
 %   \item Figure 1 is not symmetric? - because of the prior used
  %  \item What was the point of convergence (output/result)? - $x \approx (0.163,0.838)$
  %  \item Is this with $B=2$ or $B=1$ and the ensemble is split? - with $B =2$
   % \item Why this type of prior? - Otherwise utility plot has maximum only in the middle if only squared exp kernel used for prior
%\end{itemize}
%}

\begin{figure}[h!]
    \centering
    \includegraphics[width=0.9\linewidth]{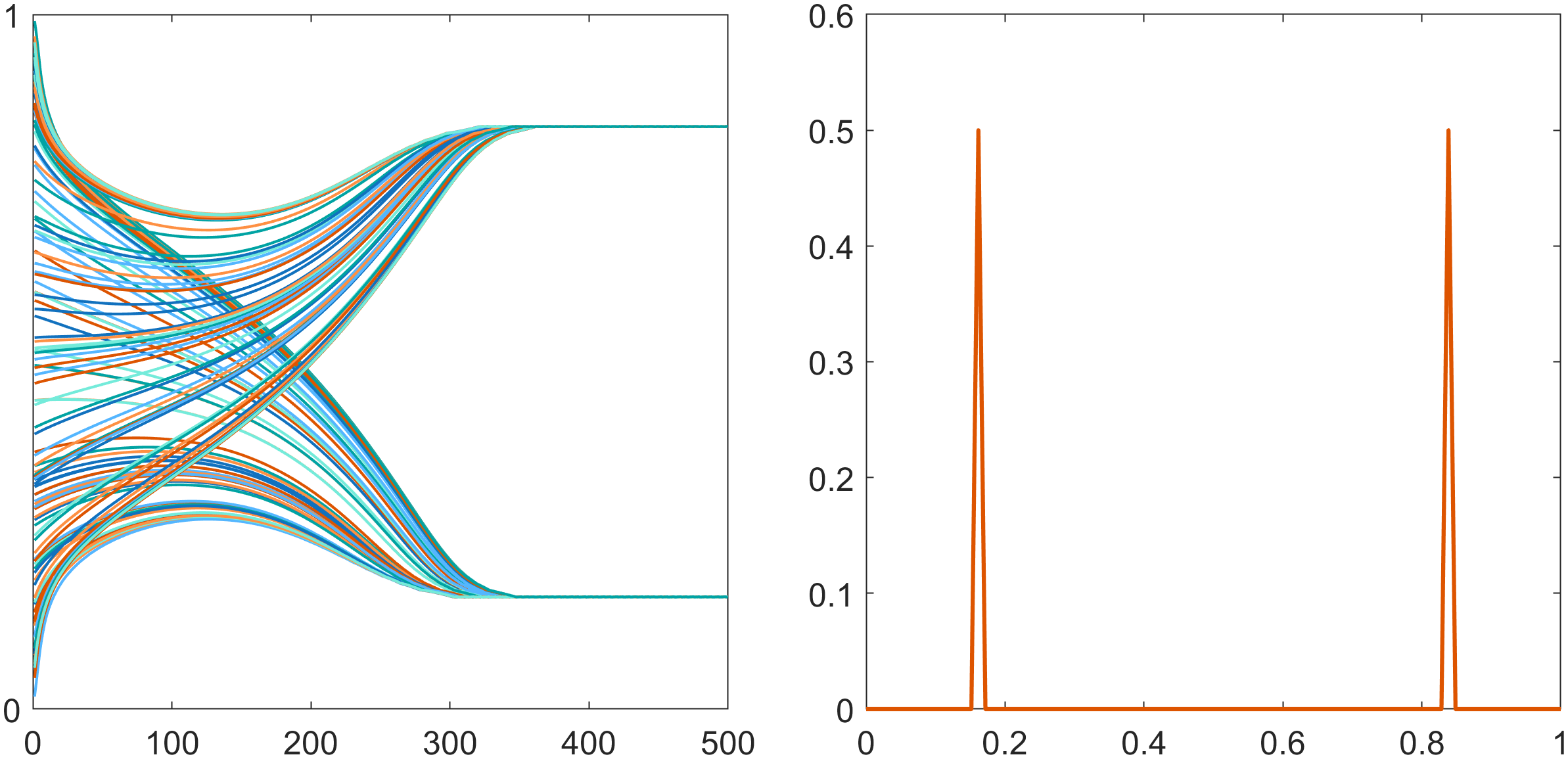}
    \caption{Trajectories of the particles through iterations on the left panel and the resulting empirical distribution of $\mu$ on the right panel is illustrated in Section \ref{subsec:poisson}.}
    \label{fig:ex1 trajectories and distb}
\end{figure}

\newpage
\subsection{Time-harmonic Schr\"odinger equation}
\label{subsec:schrodinger}
Consider the two-dimensional boundary value problem
\begin{equation}
\label{eq:timeharmonicschrodinger}
\begin{cases}
(-\Delta + V - \omega^2)u = f & \text{in } \Omega=(0,1)^2,\\
u = 0 & \text{on } \partial\Omega,
\end{cases}
\end{equation}
where $V\in L^\infty(\Omega)$ is real-valued and satisfies $V\ge 0$ a.e. in $\Omega$.
By the Rayleigh--Ritz (min--max) principle, $\lambda_1(-\Delta+V) \geq \lambda_1(-\Delta)=2\pi^2$.
Hence, if $\omega^2 < 2\pi^2$, then $\omega^2$
cannot be a Dirichlet eigenvalue of $-\Delta+V$, and \eqref{eq:timeharmonicschrodinger} admits a unique solution. The inverse problem is, given $V$ and noisy point observations of $u$, to identify $f$.

The mapping properties of \eqref{eq:timeharmonicschrodinger} are classical. Following the template of Section~\ref{subsec:poisson}, our aim is to ensure that the solution belongs to a reproducing kernel Hilbert space that is continuously embedded into $C^1(\dom)$. To this end, we recall from \cite[Chapter 6.3]{evans2022partial} %{\color{red} (Check (LATER!) if Grisvard Elliptic Problems in Nonsmooth Domains is needed as a reference for the square)} 
that for any $V \in C^\infty(\dom)$ and $f\in H^s(\dom)$, $s\geq 0$, the unique solution satisfies $u\in H^{s+2}(\dom) \cap H^1_0(\dom)$. By the Sobolev embedding theorem, any $s>1$ guarantees that $H = H^{s+2}(\dom) \cap H^1_0(\dom) \hookrightarrow C^2(\dom)$ continuously.

We define the smooth potential term by setting
\begin{align*}
        V = h_\epsilon \ast W|_{\dom}, \quad \text{where} \quad 
        W(z,y) = 200\cdot \mathbf{1}\left(\left|z-\frac 12\right|\leq 0.08\right) \cdot \mathbf{1}\left(\left|y- \frac 12\right| \leq 0.08\right) \in L^2(\R^2),
\end{align*}
where $h_\epsilon$ is Gaussian kernel with a small variance $\epsilon$ convolving $W$, which is an indicator function of a cross-shaped field.

Similarly to the Poisson example, our objective is to identify the optimal configuration of sensor locations, here, for the batch size $B=4$.
The independent point observations are contaminated with zero-mean Gaussian noise with variance 0.01. In the numerical implementation of the forward mapping $A$, both the input and output domains are discretized on a regular grid with $N=60^2$ points. For continuous evaluation of the output function, the values are interpolated based on the grid values. We impose a zero-mean squared-exponential Gaussian prior with a covariance function 
\begin{equation*}
    c(x,z) = \exp\left(-\frac{|x-z|^2}{2\sigma_0^2}\right),
\end{equation*}
where $\sigma_0 = 0.5$. We simulate Algorithm~\ref{alg:1} using 100 particles initialized from a uniform distribution over $\dom$. We regularize the empirical variance with weight $\alpha=5e-4$ and apply no repulsive regularizer, i.e. $\beta=0$. Using 120 time steps of size $8e-3$, we illustrate the results in Figure~\ref{fig:2d_ex}. First, the expected utility for a single observation is mapped on the left-hand side of the image with the reconstructed optimal sensor locations. The right-hand side plot visualizes the evolution of the ensembles towards these optimal locations.

\begin{figure}[h!]
    \centering
    \includegraphics[width=1\linewidth]{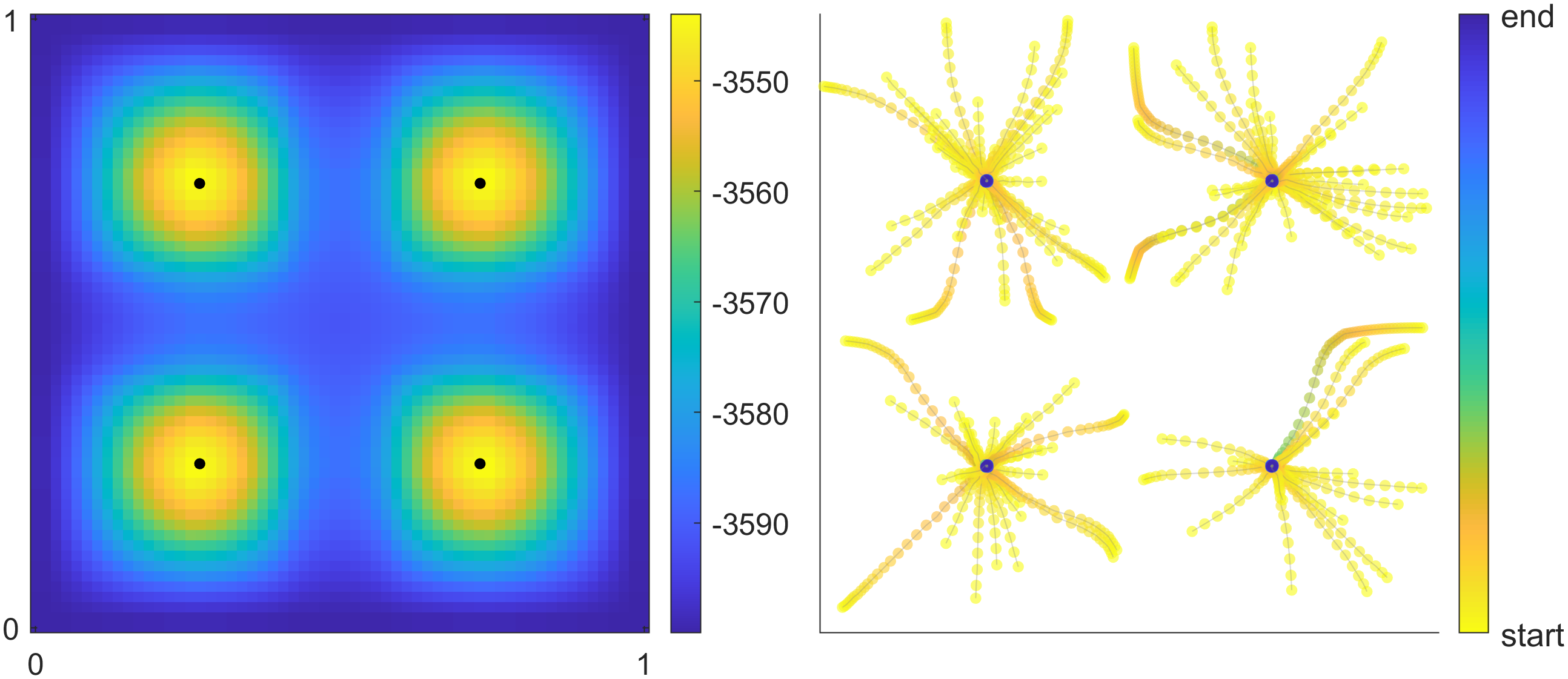}
    \caption{The expected utility for single observation and the evolution of the particle ensembles in Section \ref{subsec:schrodinger} is illustrated. The black markers on the left plot show the optimal configuration found. The plot on the right visualizes the evolution of the particles from the initial set of particles to the optimal locations. }
    \label{fig:2d_ex}
\end{figure}

\subsection{Sensitivity of regularization}
\label{subsec:sensitivity}

The regularization of the particle gradient flow relies on two distinct terms: the regularizer $\mathcal{R}_v$, which attracts particles within the same ensemble, and the regularizer $\mathcal{R}_r$, which enforces separation between two different ensembles. The best choice of the corresponding regularization parameters is influeced by various factors including the underlying model, the discretization, the number of design variables, and the time step $dt$, and, below, our aim is to illustrate these effects numerically.

\subsubsection{Regularization with $\mathcal{R}_v$}
\label{subsec:alpha sensitivity}
Consider the Poisson source identification problem in Section \ref{subsec:poisson} for batch-size $B=3$. We use the same parameters as in Section \ref{subsec:poisson} with the exception that we use 300 time steps of size $1e-5$ and initialize 900 particles from the uniform distribution such that $\mu_N^1\sim \mathcal{U}(0,\frac{1}{3})$, $\mu_N^2\sim \mathcal{U}(\frac{1}{3},\frac{2}{3})$ and $\mu_N^3\sim \mathcal{U}(\frac{2}{3},1)$.
Figure~\ref{fig:alpha_sensitivity} (left) displays the magnitude of the regularization for several parameter choices. For sufficiently large values, the measure converges to the desired number of locations before the regularization becomes dominant. In the empirical distribution plot, we observe that for smaller values of parameter $\alpha$ the particle gradient flow has not yet converged to the desired number of optimal measurement locations. To isolate the effect of $\alpha$, the parameters for the time step and its size are chosen so that the particle flow has not yet fully converged to the optimum. This is done since for this problem, when $B>2$ the gradient flow would converge to  the optimal measurement locations of which some would overlap, without the regularizer $\mathcal{R}_r$, thus making the effect of only $\alpha$ not apparent.
 
%{\color{red} (Needs more clear explanations)}

% In the plot of empirical distribution we observe that with smaller values of $\alpha$ the particle gradient flow does not yet converge to the desired number of optimal measurement locations. Note that in order to demonstrate the effect of $\alpha$, the parameters $dt$ and $T$ are chosen such that the particle flow has not yet converged to the optimal locations.

\begin{figure}[h!]
    \centering
    \includegraphics[width=\linewidth]{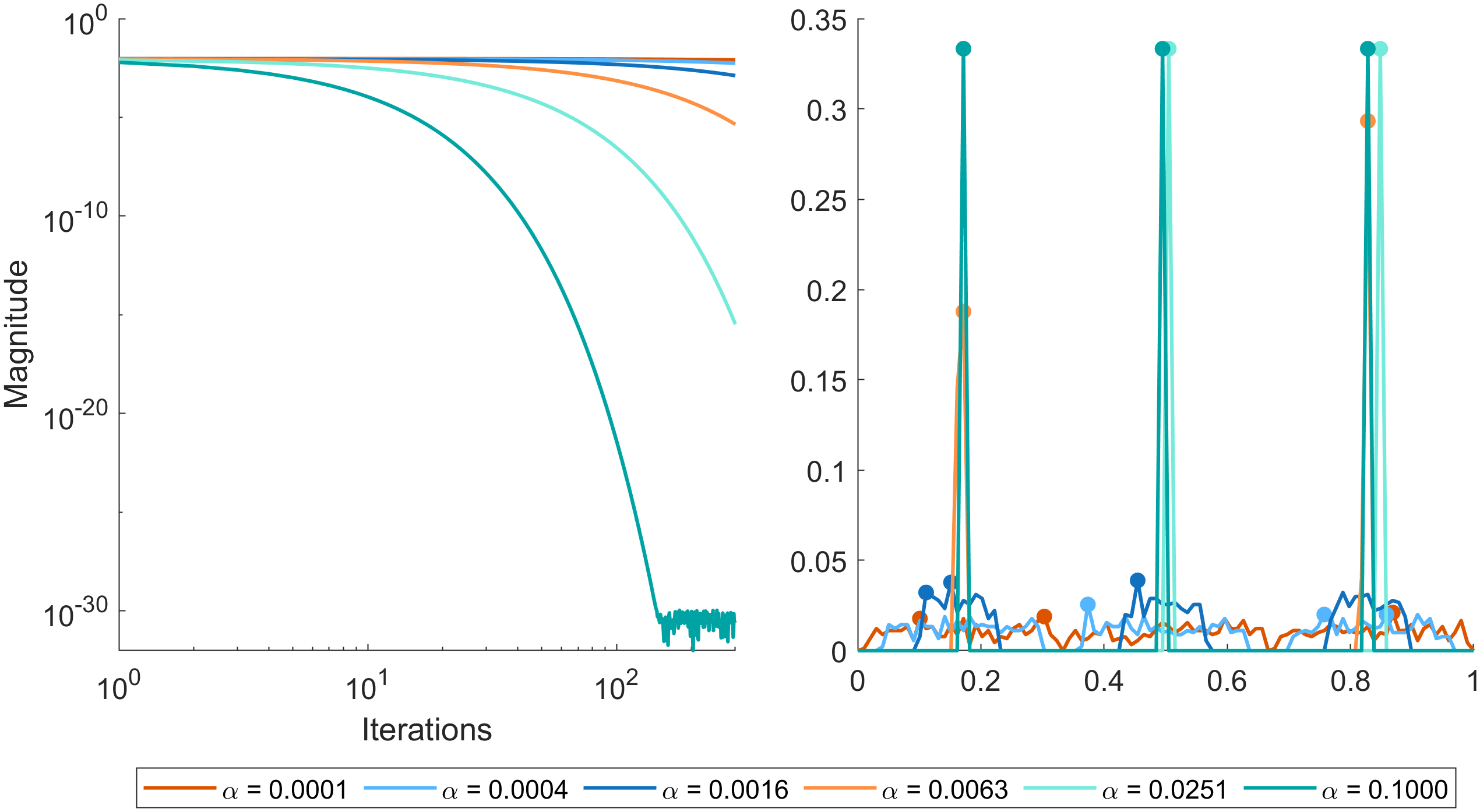}
    \caption{Effect of $\mathcal{R}_v$ illustrated in Section \ref{subsec:alpha sensitivity}. The empirical variances of the regularizer $\mathcal{R}_v$ on a log-scale are plotted on the left panel. On the right panel the resulting empirical distributions are plotted.}
    \label{fig:alpha_sensitivity}
\end{figure}

\subsubsection{Regularization with $\mathcal{R}_v$ and $\mathcal{R}_r$}
\label{subsec:beta sensitivity}
As the number of design parameters increases, the ensembles corresponding to different design variables may begin to collapse into one another. While such configurations may be optimal from an inference or information–gain perspective, practical constraints, such as the need to enforce distinct sensor locations, can make them undesirable or infeasible. To enforce separations of design variables, a repulsive regularization term can be employed in addition to $\mathcal{R}_v$. The choice of the corresponding regularization parameter $\beta$ depends on the kernel used in $\mathcal{R}_r$ and its parameters; in practice, however, $\beta$ is typically small, in comparison to the parameter $\alpha$. For instance $\beta<10^{-4}$, for the Poisson problem in Section \ref{subsec:poisson}.

In this example, we take $\mathcal{R}_r$ to be the MMD distance associated with a Gaussian kernel of standard deviation $\sigma=0.009$. We examine the effect of $\beta$ using the same one–dimensional Poisson problem as before. The batch size is set to $B=8$, and the regularization parameter for $\mathcal{R}_v$ is fixed to $\alpha=0.05$. We use 1000 time steps of size $4e-3$ and initialize 600 
particles similarly as in the example \ref{subsec:alpha sensitivity}.
The left panel of Figure~\ref{fig:beta sensitivity} shows the pairwise distances $d = (\overline{x}_j-\overline{x}_l)^2, j,l = 1,...,B,$ $j\neq l$ between the particle ensembles when the repulsive term $\mathcal{R}_r$ is applied. The right panel compares the resulting empirical distributions, from which it is evident that when $\beta$ is chosen too small the particle clouds collapse to only six optimal locations.

% With larger number of design parameters the ensembles corresponding to different design variables mcan begin to collapse together. To prevent this, repulsive regularization may be used in addition to $\mathcal{R}_v$. The values of the regularization parameter $\beta$ depend on the kernel used in the regularizer and its parameters, although the values for $\beta$ tend to be rather small, e.g. $\beta < 1e-4$.
% In this example, we utilize MMD distance associated to the Gaussian kernel with standard deviation $\sigma = 0.009$ for $\mathcal{R}_r$.
% We test the effect of $\beta$ with the same 1-dimensional Poisson's problem as previously. The batch-size is set to $B=8$ and the regularization parameter for $\mathcal{R}_v$ is fixed to $\alpha = 0.05$. Other parameters used are found in Table \ref{tab:sensitivity param}.\\
% The plot on the left of Fig. (\ref{fig:beta sensitivity}) shows the distances between different particle clouds, when $\beta R_r$ is applied. In the plot of the right of Fig. \ref{fig:beta sensitivity} a comparison of the empirical distributions achieved is shown, from which it can be seen that the particle clouds have collapsed to only 4 optimal locations, when $\beta$ is too small. 

\begin{figure}[h!]
    \centering
    \includegraphics[width=\linewidth]{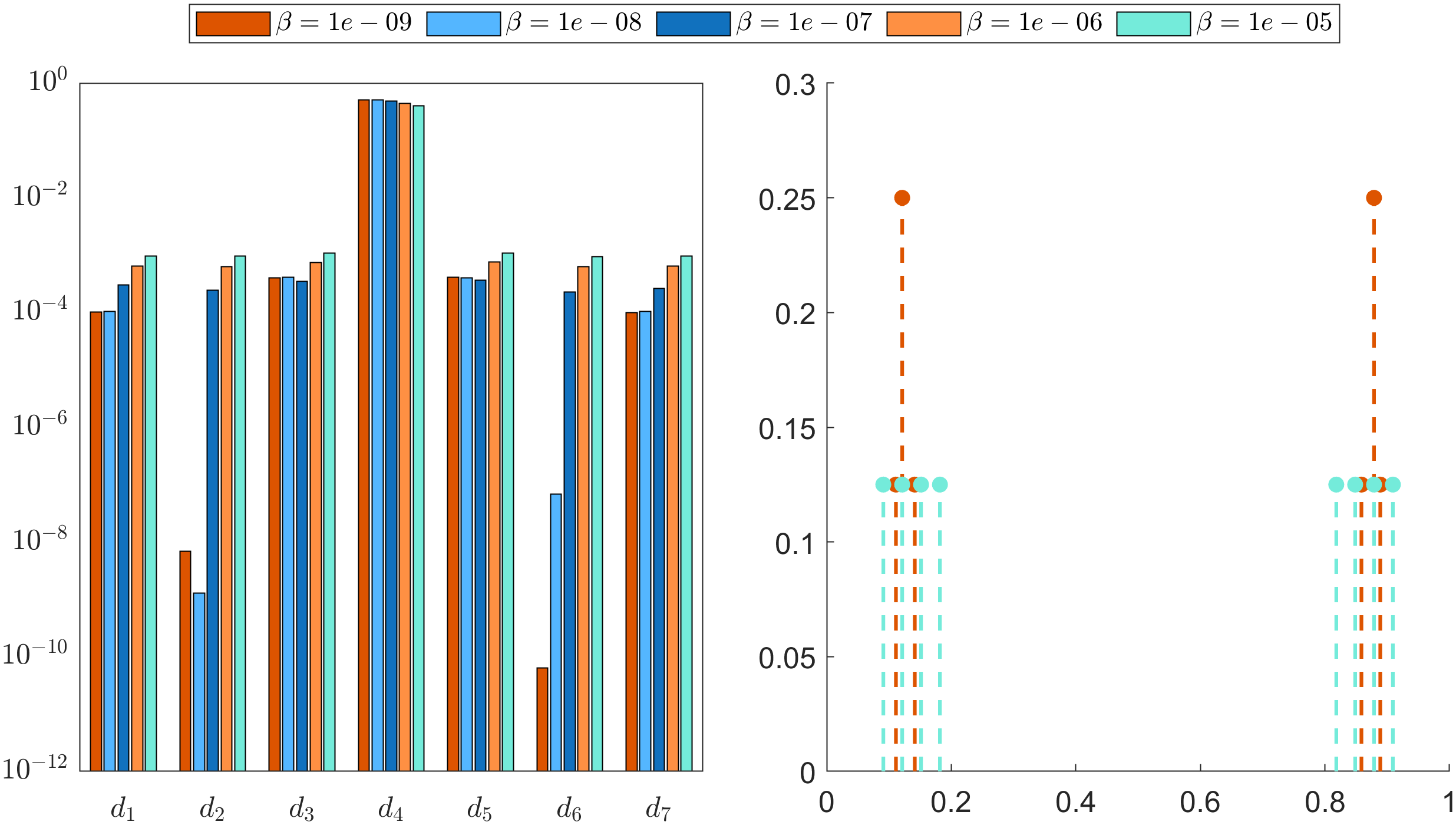}
    \caption{Effect of $\mathcal{R}_r$ illustrated in Section \ref{subsec:beta sensitivity}. The distances $d_1(\mu_N^1,\mu_N^2),...,d_7(\mu_N^7,\mu_N^8)$ between particle clouds on a log-scale are plotted on the left panel. On the right panel the comparison of two resulting empirical distributions are plotted.}
    \label{fig:beta sensitivity}
\end{figure}

\newpage
\section{Conclusion}

In this work, we established a rigorous connection between the frequentist batch-based A-optimal design problem formulated over measures and a corresponding Bayesian inference framework. By interpreting the relaxation to positive measures of fixed mass as arising from a non-parametric Gaussian observation model, we derived an infinite-dimensional expected utility functional and proved its concavity with respect to the design measure. This result provides a principled Bayesian justification for the convex relaxation previously proposed in the frequentist setting and clarifies the statistical meaning of optimizing over measures rather than discrete sensor configurations. The derivation of first variations and Wasserstein gradients further enabled the formulation of a particle-based optimization scheme grounded in gradient flow dynamics.

Beyond the theoretical characterization, we proposed a regularized tensorized formulation that lifts the optimization back to the product space while retaining the advantages of the relaxed representation. The introduced variance and maximum mean discrepancy penalties offer practical mechanisms to promote concentration of individual design measures and enforce separation between distinct sensor ensembles. Together, these elements yield a flexible framework for batch-based Bayesian experimental design over continuous domains. Future work may investigate convergence properties in greater generality, extensions to other optimality criteria, and applications to high-dimensional inverse problems arising in large-scale scientific and engineering systems.

\section*{Acknowledgement}

This work was supported by the Research Council of Finland (decisions 353094, 348504, 359183). SM was supported by Emil Aaltonen foundation. 

\bibliographystyle{plain}
\bibliography{references}

\appendix

\end{document}